\pdfoutput=1
\documentclass[fleqn,11pt,letterpaper]{article}
\usepackage{amsfonts}
\usepackage{amssymb}
\usepackage{amsmath}
\usepackage{amsthm}
\usepackage{enumitem}
\usepackage{multirow}
\usepackage{float}
\usepackage{graphicx}
\usepackage{color}
\usepackage[margin=2.57cm]{geometry}
\definecolor{darkred}{rgb}{0.5,0.2,0.2}
\usepackage[pdftitle={Decoupling the short- and long-term behavior of stochastic volatility},pdfauthor={Mikkel Bennedsen, Asger Lunde and Mikko Pakkanen},colorlinks=TRUE,allcolors=darkred]{hyperref}
\usepackage{booktabs}
\usepackage{pdflscape}
\usepackage[round]{natbib}
\usepackage{tabularx}

\usepackage[export]{adjustbox}

\usepackage{colortbl}
\definecolor{darkgrey}{rgb}{0.7,0.7,0.7}
\newcommand{\MCSone}{\cellcolor{darkgrey}}  

\definecolor{grey}{rgb}{0.9,0.9,0.9}
\definecolor{purple}{rgb}{0.6,0,0.6}
\newcommand{\MCStwo}{\cellcolor{grey}}  

\usepackage{float}
\floatstyle{plaintop}
\restylefloat{table}

\newcommand\mikko[1]{\textcolor{black}{#1}}
\newcommand\mikkof[1]{\textcolor{black}{#1}}

\newcommand\mikkol[1]{\textcolor{black}{#1}}

\newcommand\mikkel[1]{\textcolor{black}{#1}}
\newcommand\mik[1]{\textcolor{black}{#1}}

\newcommand\mikkell[1]{\textcolor{black}{#1}}

\newcommand\mikkelll[1]{\textcolor{black}{#1}}
\newcommand\miksim[1]{\textcolor{black}{#1}}

\newcommand\mikkolll[1]{\textcolor{black}{#1}}

\restylefloat{table}

\theoremstyle{plain}
\newtheorem{theorem}{Theorem}[section]

\newtheorem{proposition}{Proposition}[section]

\theoremstyle{definition}
\newtheorem{example}{Example}[section]

\theoremstyle{remark}
\newtheorem{remark}{Remark}[section]

\def\F{{\mathcal F}}

\newcommand{\N}{\mathbb{N}}

\newcommand{\R}{\mathbb{R}}

\newcommand{\E}{\mathbb{E}}
\newcommand{\BSS}{\mathcal{BSS}}

\def\bi{\begin{itemize}}
\def\ei{\end{itemize}}

\allowdisplaybreaks
\graphicspath{{Figures/}}
\numberwithin{equation}{section}

\newif\ifi

\begin{document}

\title{Decoupling the short- and long-term behavior of stochastic volatility}


\author{Mikkel Bennedsen\thanks{
Department of Economics and Business Economics and CREATES, 
Aarhus University, 
Fuglesangs All\'e 4,
8210 Aarhus V, Denmark.
E-mail:\
\href{mailto:mbennedsen@econ.au.dk}{\nolinkurl{mbennedsen@econ.au.dk}}. 
} \and 
Asger Lunde\thanks{
Department of Economics and Business Economics and CREATES, 
Aarhus University, 
Fuglesangs All\'e 4,
8210 Aarhus V, Denmark.
E-mail:\ 
\href{mailto:alunde@econ.au.dk}{\nolinkurl{alunde@econ.au.dk}}. 
} \and 
Mikko S. Pakkanen\thanks{
Department of Mathematics, 
Imperial College London, 
South Kensington Campus,
London SW7 2AZ, UK and CREATES, Aarhus University, Denmark.
E-mail:\ 
\href{mailto:m.pakkanen@imperial.ac.uk}{\nolinkurl{m.pakkanen@imperial.ac.uk}}.
}}

\maketitle
\begin{abstract}
We introduce a new class of continuous-time models of the stochastic volatility of asset prices. The models can simultaneously incorporate roughness and slowly decaying autocorrelations, including proper long memory, which are two stylized facts often found in volatility data. Our prime model is based on the so-called Brownian semistationary process and we derive a number of theoretical properties of this process, relevant to volatility modeling. Applying the models to realized volatility measures covering a vast panel of assets, we find evidence consistent with the hypothesis that time series of realized measures of volatility are both rough and very persistent. Lastly, we illustrate the utility of the models in an extensive forecasting study; we find that the models proposed in this paper outperform a wide array of benchmarks considerably, indicating that it pays off to exploit both roughness and persistence in volatility forecasting.
\end{abstract}


\vspace*{1em}


\vspace*{1em}


\baselineskip=17pt

\section{Introduction}\label{sec:intro}

A large body of literature in financial econometrics has found that the volatility of asset prices, or realized measures thereof, is persistent, in the sense that its autocorrelation function decays very slowly \citep[e.g.][]{DGE1993,Baillie1996,Cont2001,ABDL01,ABDL03,BCDS2019}. For this reason, it has become popular to employ stochastic volatility models that enjoy the long memory property, meaning that their autocorrelation function is not integrable \citep[e.g.][]{PG2003,FK2011}. The empirical evidence of persistence inspired \cite{CR96} to introduce the \emph{fractional stochastic volatility} (FSV) model, which is based on  a fractional Ornstein--Uhlenbeck process, driven by a fractional Brownian motion (fBm) with Hurst parameter $H > 1/2$, resulting in a continuous-time model of stochastic volatility with the long memory property.

Recently, there has been considerable interest in rough models of volatility. This is due to both theoretical developments in implied volatility modeling \citep[][]{ALV2007,fukasawa15} as well as empirical evidence based on realized volatility \citep[][]{GJR14}. An important contribution to this literature is the model suggested by \cite{GJR14}, which is inspired by FSV model of \cite{CR96}. While the FSV model of \cite{CR96} uses an fBm with Hurst index $H > 1/2$, giving rise to long memory, \cite{GJR14}, by contrast, switch to an fBm with $H< 1/2$ to allow for roughness. They therefore term their model the \emph{rough fractional stochastic volatility} (RFSV) model.

When designing a realistic model of volatility, it is important to be aware of the shortcoming, stressed by \citet[][]{GS04}, of \emph{self-similar processes}, such as the fBm, whereby their roughness is \emph{inseparable} from their long term behavior. A work-around to this limitation is to use fBm as a driver of a stochastic differential equation (SDE), as in the fractional Ornstein-Uhlenbeck process used in the FSV and RFSV models. Setting aside the intricacies of working with SDEs driven by fBm, this approach does not make it easy to specify the long-term behavior of the process independent of roughness, however. A case in point, the Hurst index $H$ still controls the rate of decay of the autocorrelation function of the fractional Ornstein-Uhlenbeck process in the rough regime $H\in \left( 0, \frac{1}{2} \right)$.

The first contribution this paper is to introduce a class of continuous-time models of volatility incorporating both roughness (irregular behavior at short time scales) and persistence (strong dependence at longer time scales). In particular, we advocate the use of \emph{Brownian semistationary} ($\BSS$) processes \citep{BNSc07,BNSc09} as models of logarithmic volatility. As we will show, these processes are flexible in the sense that they allow for decoupling of the fine properties (roughness) from the long-term behavior (memory/persistence). Indeed, we show that $\BSS$ processes, under suitable specifications, can accommodate both bona fide long memory (i.e., non-integrable autocorrelations) and short memory (exponentially decaying autocorrelations), while the degree of roughness is specified independently. Under rather general conditions, $\BSS$ processes are stationary, and they allow for easy inclusion of non-Gaussianity and the leverage effect. In spite of their generality, the suggested models are simple, in terms of their mathematical structure, and parsimonious, relying on only two parameters controlling short- and long-term behavior, respectively. Moreover, as we demonstrate below, the models are easy to estimate, simulate, and forecast.
 
The second contribution of the paper is to develop an estimation procedure for the parameters of the proposed models. Since the stochastic volatility process of an asset price is unobservable, any estimation procedure needs to be based on indirect measures of volatility, which might constitute (very) noisy measurement of the true underlying process. This engenders a measurement-error problem, which can bias estimates of the model parameters, especially estimates of the roughness index \citep{bennedsen16a}. We propose noise-robust estimators of all parameters in the proposed models. Our robust estimator of the roughness index is novel and relies on a non-linear least squares (NLLS) regression. To our knowledge, this is first time such a robust estimator has been considered in a study of rough volatility.\footnote{Concurrently with the writing of the present version of this paper, \cite{FTW2019} \mikkolll{have} developed a quasi-likelihood approach to estimation of rough volatility models, which also seeks to take account of the noise induced by using a realized volatility measure in place of the latent volatility process.} In a simulation study, we inspect the finite sample properties of the robust estimator of the roughness index and compare with those of a basic (non-robust) estimator. We find that the basic estimator, which does not take measurement noise into account, can become severely biased when the data are noisy, while our novel NLLS-based estimator can significantly alleviate this bias.

Whereas the first two contributions of the paper are methodological, the third is empirical. We apply our modelling framework to the study of data on the E-mini S\&P $500$ futures contract at a wide range of intraday time scales, ranging from ten minutes to one day. To our knowledge, this is the first time the roughness properties of intraday volatility have been studied, a surprising fact given that roughness is inherently a fine scale property and is thus most accurately studied at small scales. We also study the roughness and persistence of daily realized volatility measures of close to $2\ 000$ individual equities. Our empirical findings indicate that both intraday and daily volatility measures display high degrees of roughness and high degrees of persistence, including possible long memory, thus indicating that our proposed theoretical models could potentially be useful as models of stochastic volatility. 

Lastly, to examine the practical usefulness of the suggested volatility models, we conduct an extensive forecasting study. We find that our proposed models outperform a wide array of benchmark models, especially at intraday time scales. This indicates that it is important to carefully model both small- and large-scale behavior of stochastic volatility in forecasting and, therefore, that models that decouple the behavior at these scales can be useful in practice.

The rest of the paper is structured as follows. Section \ref{sec:model} defines the concepts of roughness and persistence, in the context of this paper, through the autocorrelation of a stochastic process. The section then introduces a class of stochastic volatility models, based on $\mathcal{BSS}$ processes, that are able to parsimoniously incorporate roughness and high persistence, including long memory, simultaneously. Section \ref{sec:estimate} describes the estimation procedure and presents the novel noise-robust estimator of the  roughness parameter. Section \ref{sec:sim} contains a number of simulation studies, investigating the finite sample properties of the proposed estimator. Section \ref{sec:empirical} is empirical and estimates the various models using realized volatility data. Section \ref{sec:forecast} presents a forecasting study, where we compare our new models to existing volatility forecasting models. Lastly, Section \ref{sec:concl} concludes. Some further details and proofs of the technical results are given in an Appendix. A Web Appendix, containing further simulation experiments  is available online.


\section{Models of stochastic volatility that decouple short- and long-term behavior}\label{sec:model}
Let $X = (X_t)_{t \in \R}$ be a stationary continuous stochastic process, defined on a probability space $(\Omega, \mathcal{F}, \mathbb{P})$, satisfying the usual conditions. Define $\rho$ as the autocorrelation function (ACF) of $X$, i.e.
\begin{align*}
\rho(h) := \textnormal{Corr}(X_{t+h},X_t), \quad h \in \R.
\end{align*}

We say that $X$ is \emph{rough} if its ACF adheres to the asymptotic relationship
\begin{align}\label{eq:assA}
1-\rho(h)  \sim c |h|^{2\alpha + 1}, \quad |h| \rightarrow 0,
\end{align}
for a constant $c>0$ and some $\alpha \in \left(-\frac{1}{2},\frac{1}{2}\right)$.\footnote{Here and below ``$\sim$" indicates that the ratio between the left- and right-hand side tends to one --- when the constant in this relationship is immaterial, we will denote it by the generic $c$, which  may vary from from formula to another. } We call $\alpha$ the \emph{roughness index} of $X$. For stationary Gaussian processes, the relationship \eqref{eq:assA} implies that the process $X$ has a modification with locally $\phi$-H\"older continuous trajectories for any $\phi \in (0,\alpha+1/2)$  \citep[e.g.][Proposition 2.1]{bennedsen16a}. \mikkolll{For completeness, we recall here that a function $x : \R \rightarrow \R$ is locally $\phi$-H\"older continuous for $\phi \in (0,1]$ if for any compact interval $I \subset \R$ there is a constant $c_I>0$ such that
\begin{align*}
|x(s)-x(t)| \leq c_I |s-t|^{\phi}, \quad s,\,t \in I.
\end{align*}}Thus, the index $\alpha$ can be seen as a measure of roughness, with small values indicating more roughness. It is worth recalling here that a standard Brownian motion has locally $\phi$-H\"older continuous trajectories for any $\phi \in (0,1/2)$. Thus, negative values of $\alpha$ suggest trajectories rougher than those of a standard Brownian motion.

Rough models of volatility are consistent with some empirically observed features of implied volatility surfaces \citep{gatheral06}. In particular, as shown in \cite{ALV2007} and \cite{fukasawa15}, such models can accurately capture the short-time behavior of the at-the-money volatility skew, which conventional local/stochastic volatility models based on It\^o diffusions fail to capture. To model the roughness of volatility, earlier studies have mainly relied on the ``canonical" rough process, the fractional Brownian motion (fBm) with Hurst index $H \in (0,1/2)$. For the fBm, the simple relationship $H = \alpha + 1/2$ holds, which means that $H < 1/2$ implies roughness.

That volatility is very persistent has long been a well-established fact \cite[e.g.,][]{BW00,ABDL03}. A large body of literature has therefore focused on modeling (log) volatility using persistent models, i.e., volatility models whose autocorrelations decay at a slow polynomial rate:
\begin{align}\label{eq:assB}
\rho(h)  \sim c |h|^{-\beta}, \quad |h| \rightarrow \infty,
\end{align}
for a constant $c>0$ and some $\beta \in (0,\infty)$. When $\beta \in (0,1)$, the correlation function $\rho$ is not integrable, that is, $\int_0^\infty |\rho(h)|dh = \infty$, and we say that $X$ has the \emph{long memory property}. Models of log volatility that are able to reproduce the long memory property have been traditionally built using an fBm with Hurst index $H = 1-\beta/2 \in (1/2,1)$ as  driving noise; see, e.g., \cite{CR96}, \cite{comte96}, \cite{CR98}, and \cite{CCR12} for literature on continuous-time volatility models with long memory.

\mikkolll{The fBm is stationary only in terms of increments. While for small $H<1/2$, the process has highly anti-persistent increments, rendering its trajectories seemingly stationary at moderately long time scales, the formal lack of stationarity of levels is a nuisance when we would like to talk about long-term properties of volatility like \eqref{eq:assB}. The remedy, proposed by \cite{CR96,CR98} is to construct a volatility model from a fractional Ornstein-Uhlenbeck process, i.e., an Ornstein-Uhlenbeck process driven by an fBm \citep{CKM2003} with Hurst index $H$. The choice $H>1/2$ advocated by \cite{CR96,CR98} gives rise to a model with long memory, the fractional stochastic volatility (FSV) model, while the choice $H < 1/2$, championed recently by \cite{GJR14}, results in a model with rough trajectories, the rough fractional stochastic volatility (RFSV) model.} 

\mikkolll{The theory of fractional Ornstein-Uhlenbeck processes, Theorem 2.3 of \citet{CKM2003} to be specific, tells us that the processes have an autocorrelation function that decays at a polynomial rate with exponent $\beta=2-2H$. Thus, for $H<1/2$ the process loses the long memory property, as formalized above, while it is still undoubtedly persistent. However, letting the parameter $H$ govern both the roughness and long-term behaviour of the process via $\alpha = H + 1/2$ and $\beta = 2-2H$, respectively, appears to have neither theoretical nor empirical motivation in the context of volatility modeling. This prompts us look for processes $X$ for which we can specify the roughness index $\alpha$ independently of the memory parameter $\beta$, or vice versa, or even replace polynomial decay of autocorrelations with exponential decay without affecting roughness.}
Given a process $X$ with such \mikkolll{qualities}, we then define our \mikkolll{general} model for the spot volatility process $\sigma_t$ as
\begin{align}\label{eq:sig}
\sigma_t = \xi \exp\left(X_t \right), \quad t \geq 0,
\end{align}
where $\xi > 0$ is a free parameter. 


\subsection{Models for log volatility}
We \mikkolll{entertain} two main candidates for $X$: the Cauchy process and the so-called Brownian semistationary process. Common to the two processes is that they, in the setting we consider here, both have two parameters, controlling the short- and long-term behavior, respectively. Thus, the number of parameters is, in either case, no higher than in the RFSV model \citep[][Sect. 3.1]{GJR14}. We will see that the Brownian semistationary process, in particular, is ideally suited to volatility modeling

\subsubsection{The Cauchy process}\label{sec:cauchy}
A flexible Gaussian process that decouples the short- and long-term behavior can be obtained by using the \emph{Cauchy class} of autocorrelation functions \citep{GS04}. The resulting \emph{Cauchy process} is a centered, stationary Gaussian process $G = (G_t)_{t\in \R}$ with autocorrelation function
\begin{align*}
\rho(h) = \left(1+|h| ^{2\alpha+1}\right)^{-\frac{\beta}{2\alpha+1}}, \quad h \in \R,
\end{align*}
where $\alpha \in \left(-\frac{1}{2},\frac{1}{2}\right)$ and $\beta > 0$. In particular, the process $G$ satisfies \eqref{eq:assA} and \eqref{eq:assB}. 

However, the \mikkolll{main} limitation of this process, from a modeling point of view, is its inherent Gaussianity. Yet it is possible to go beyond Gaussianity by volatility modulation, that is, by specifying a process
\begin{align}\label{eq:volCauchy}
X_t = \int_0^t v_s dG_s, \quad t \geq 0,
\end{align}
using the Cauchy process $G$ and some additional process $v = (v_t)_{t\in \R}$ that models the volatility of volatility. Under suitable assumptions, $X$ will inherit the roughness properties of $G$, see, e.g., \citet[Example $3$]{BNCP09}. It is worth stressing that since $G$ is typically a non-semimartingale --- and with parameter values relevant to rough volatility modeling it indeed is --- the stochastic integral in \eqref{eq:volCauchy} cannot be defined as an \emph{It\^o integral}, but \emph{pathwise Young integration} \citep{DN11} needs to be used, which requires some additional assumptions on the roughness of $v$. In the case where the Cauchy process $G$ is rough, these assumptions become rather restrictive, unfortunately --- for example, $v$ cannot be a semimartingale (with non-zero quadratic variation).

In sum, the Cauchy class provides a convenient model that decouples roughness and memory properties. 
Since the Cauchy class is characterized by a closed-form ACF, it is very easy to work with. Moreover, as we will see in the forecasting experiment below, volatility forecasts derived using a Cauchy class model perform rather well. However, this class of models is not easy to extend beyond Gaussianity, which limits its applicability as a model of log volatility. Next, we propose a different modeling framework, \mikkolll{which easily accommodates non-Gaussianity.}

\subsubsection{The Brownian semistationary process}\label{sec:bss}
A stochastic process that is able to capture all of our desiderata, consisting of roughness, strong persistence, stationarity, and non-Gaussianity is the Brownian semistationary process ($\BSS$), which was introduced in \cite{BNSc07,BNSc09}. This process is defined via the moving-average representation 
\begin{align}\label{eq:BSS}
X_t = \int_{-\infty}^t g(t-s) v_s dW_s, \quad t \geq 0,
\end{align}
where $W = (\mikkolll{W_t})_{t \in \R}$ is a standard Brownian motion defined on $\mathbb{R}$, $g:(0,\infty) \rightarrow \R$ is a square-integrable kernel function, and $v = (v_t)_{t \in \mathbb{R}}$ is an adapted, covariance-stationary volatility (of volatility) process. Note that when $v$ is deterministic, $X$ is Gaussian, while a stochastic $v$ makes $X$ non-Gaussian. In particular, when $v$ is independent of $W$, we have
\begin{align*}
X_t | (v_s)_{s\leq t} \sim N\left(0,\int_0^{\infty}g(x)^2v^2_{t-x}dx \right),
\end{align*}
showing that the marginal distribution of $X_t$ is a normal mean-variance mixture with conditional variance governed by $v$ and $g$. It is shown in \cite{BNBV13} that when $g$ is given by the so-called gamma kernel (Example \ref{ex:gamma} below), we can choose the normal-inverse Gaussian (NIG) distribution as the marginal distribution of $X$.\footnote{In an earlier version of this paper (\url{https://arxiv.org/abs/1610.00332v2}), we showed that the NIG distribution appears to fit the empirical distribution of log volatility very well, especially at intraday time scales. This is an encouraging property of the $\BSS$ framework, and using this as a guide to specify a model for volatility of volatility is a promising approach, but beyond the scope of the present study.}

Under the assumptions given above, the process $X$ is already well-defined and covariance-stationary. For stationarity, integration from $-\infty$ in \eqref{eq:BSS} is crucial. We will now introduce additional assumptions concerning the properties of the kernel function $g$, which enable us to derive some theoretical results \mikkolll{on} $X$. 
\begin{enumerate}[label=(A\arabic*),ref=A\arabic*,leftmargin=3em]
\item\label{A:zero} For some $\alpha \in (-1/2,1/2)\setminus\{ 0\}$,\begin{align}
g(x)  &= x^{\alpha} L_0(x), \quad x \in (0,1], \label{eq:g1}
\end{align}
where the function $L_0$ is continuously differentiable, bounded away from zero, and slowly varying function at zero in the sense that $\lim_{x \downarrow 0} \frac{L_0(tx)}{L_0(x)} = 1$ for all $t>0$.\footnote{We refer to \cite{BGT} for an extensive treatment of slowly varying (and regularly varying) functions.} Furthermore, the derivative $L'_0$ of $L_0$ satisfies
\begin{align*}
|L'_0(x)| \leq C(1+x^{-1}), \quad x \in (0,1],
\end{align*}
for some constant $C>0$.
\item\label{A:smooth} The function $g$ is continuously differentiable with derivative $g'$ that is ultimately monotonic and satisfies $\int_1^\infty g'(x)^2 dx<\infty$.
\item\label{A:infty} For some $\lambda \geq 0$ and $\gamma \in \R$, such that $\gamma > 1/2$ when $\lambda = 0$,
\begin{align}
g(x)  = e^{-\lambda x} x^{-\gamma} L_1(x), \quad x \in (1,\infty), \label{eq:g2}
\end{align}
where $L_1$ is slowly varying at infinity and bounded away from zero and $\infty$ on any finite interval.
\end{enumerate}

Assumptions \eqref{A:zero} and \eqref{A:infty} refine the earlier standing assumption that $g$ is square-integrable. Indeed, in the case $\lambda=0$, a simple application of the so-called Potter bounds \citep[][Theorem 1.5.6(ii)]{BGT} shows that $\alpha > -1/2$ and $\gamma > 1/2$ are sufficient conditions for $g$ to be square integrable under the specifications \eqref{eq:g1} and \eqref{eq:g2}. Similarly, in the case $\lambda>0$, the conditions $\alpha > -1/2$ and $\gamma \in \R$, under \eqref{eq:g1}  and \eqref{eq:g2}, suffice for square integrability. The assumptions \eqref{A:zero}, \eqref{A:smooth}, and \eqref{A:infty} are similar to those used in \cite{BLP15}, the only difference being that \eqref{A:infty} is slightly more specific compared to the corresponding assumption in \mikkel{that paper}.

The following proposition shows that under assumptions \eqref{A:zero} and \eqref{A:smooth}, the $\BSS$ process $X$, defined by \eqref{eq:BSS}, has $\alpha$ as its roughness index in the sense of Equation \eqref{eq:assA}. The result is a straightforward adaptation of Proposition 2.1 in \cite{BLP15} and we refer to that paper for a proof. Below, and in what follows, we denote by $\rho_X$ the autocorrelation function of $X$.

\begin{proposition}\label{prop:BSSa}
If the kernel function $g$ satisfies \eqref{A:zero} and \eqref{A:smooth}, then
\begin{align*}
1-\rho_X(h) \sim \left( \frac{\frac{1}{2\alpha + 1} + \int_0^\infty \big( (x+1)^\alpha - x^\alpha\big)^2 dx}{\int_0^{\infty}g(x)^2dx} \right) L_0(|h|)^2 |h|^{2\alpha +1}, \quad |h| \downarrow 0.
\end{align*}
\end{proposition}

The tail behavior of the kernel function $g$ at infinity, as specified by \eqref{A:infty}, controls the long-term memory properties of $X$. We consider first the case $\lambda=0$, where the exponential damping factor in \eqref{eq:g2} disappears. In this case, the parameter $\gamma$ controls the asymptotic memory properties of the $\BSS$ process $X$.
\begin{proposition}\label{prop:BSSlm}
Suppose that the kernel function $g$ satisfies \eqref{A:infty}, with $\lambda = 0$, and $\int_0^1 g(x) dx <  \infty$. 
\begin{enumerate}[label=(\roman*),ref=\roman*,leftmargin=3em]
\item\label{asymp1} If $\gamma \in (1,\infty)$, then
\begin{align*}
\rho_X(h) \sim  \left( \frac{\int_0^{\infty}g(x)dx}{\int_0^{\infty}g(x)^2dx} \right) L_1(|h|) |h|^{-\gamma}, \quad |h| \rightarrow \infty.
\end{align*}
\item\label{asymp2} If $\gamma \in (1/2,1)$, then
\begin{align*}
\rho_X(h) \sim  \left( \frac{\int_0^{\infty} x^{-\gamma}(1+x)^{-\gamma}dx}{\int_0^{\infty}g(x)^2dx} \right) L_1(|h|)^2 |h|^{1-2\gamma}, \quad |h| \rightarrow \infty.
\end{align*}
\end{enumerate}
\end{proposition}

\begin{remark}
In the critical case $\gamma=1$, the asymptotic behavior of $\rho_X$ is indeterminate under \eqref{A:infty}, and would require additional assumptions on the slowly varying function $L_1$.
\end{remark}

\begin{remark}\label{rem:polyBSS}
It follows from Proposition \ref{prop:BSSlm} that if the kernel function $g$ satisfies \eqref{A:infty} with $\lambda = 0$, the resulting $\BSS$ process will --- up to a slowly varying factor --- have an autocorrelation function which decays polynomially as $h\rightarrow \infty$. Indeed, letting the rate of polynomial decay be denoted by $\beta$ as in \eqref{eq:assB}, we see that $\gamma = \beta$ when $\gamma > 1$, while $\beta = 2\gamma-1$ when $\gamma \in (1/2,1)$. From this it also follows that if $\gamma \in (1/2,1)$, then
\begin{align*}
\int_0^{\infty} \rho_X(h)dh = \infty,
\end{align*}
i.e., $X$ has the long memory property. 
\end{remark}

In contrast to the case $\lambda=0$, the assumption \eqref{A:infty} with $\lambda>0$ allows for models where autocorrelations decay to zero exponentially fast, leading to short memory, as shown in the following result.
\begin{proposition}\label{prop:BSSsm}
Suppose that the kernel function $g$ satisfies \eqref{A:infty} with $\lambda > 0$ and \mikkel{$\gamma \in \R$}, such that $\int_0^1 g(x) dx <  \infty$. Then
\begin{align*}
\rho_X(h) \sim  \left( \frac{\int_0^{\infty} g(x)e^{-\lambda x}dx}{\int_0^{\infty}g(x)^2dx} \right) e^{-\lambda |h|} |h|^{-\gamma} L_1(|h|), \quad  |h| \rightarrow \infty.
\end{align*}
\end{proposition}

\begin{remark}
Assumption \eqref{A:zero} is a sufficient condition for the requirement that $\int_0^1 g(x) dx <  \infty$ in Propositions \ref{prop:BSSlm} and  \ref{prop:BSSsm}.
\end{remark}

An example of a kernel function that satisfies \eqref{A:zero}, \eqref{A:smooth}, and \eqref{A:infty}, which will be important for us later on, is the \emph{power law kernel}.
\begin{example}[Power law kernel]\label{ex:power}
Let $g$ be the power law kernel
\begin{align}\label{eq:powerLaw}
g(x) = x^{\alpha}(1+x)^{-\gamma-\alpha}, \qquad x > 0, \qquad \alpha \in \left(-\frac{1}{2},\frac{1}{2}\right), \quad \gamma \in \left(\frac{1}{2},\infty\right).
\end{align}
\citet[Example 2.2]{BLP15} show that this kernel function indeed satisfies \eqref{A:zero}, \eqref{A:smooth}, and \eqref{A:infty}. In particular, with this kernel function, the $\BSS$ process $X$ has roughness index $\alpha$ and memory properties controlled by $\gamma$, as expounded in Proposition \ref{prop:BSSlm}. In the following, we will refer to the $\BSS$ process with the power law kernel as the \emph{Power-$\BSS$} process.
\end{example}

Later we will need the correlation structure of the Power-$\BSS$ process. By covariance-sta\-tion\-arity of $v$, the autocovariance function of the general $\BSS$ process \eqref{eq:BSS} is
\begin{align}\label{eq:cX}
c_X(h) := Cov(X_t,X_{t+h}) = \E[v_0^2] \int_0^{\infty} g(x)g(x+|h|) dx, \quad h \in \R.
\end{align}
From this we deduce that when $g$ is given as in \eqref{eq:powerLaw} we have
\begin{align*}
c_X(0) := Var(X_t) &=\E[v_0^2] \int_0^{\infty} x^{2\alpha}\mikkel{(1+x)^{-2\gamma -2\alpha}} dx \\
	&=\E[v_0^2] \mathcal{B}(2\alpha + 1, \mikkel{2\gamma} -1),
\end{align*}
where $\mathcal{B}(x,y) = \int_0^1t^{x-1}(1-t)^{y-1}dt = \int_0^{\infty} t^{x-1}(1+t)^{-x-y} dt$ is the beta function \citep[e.g.,][formula 8.380.3]{integrals}. To calculate the correlation function $\rho_X(h) = c_X(h)/c_X(0)$ we resort to numerical integration of \eqref{eq:cX}. Note that $\rho_X$ does not depend on $\E[v_0^2]$.

Another example of a kernel function that satisfies equations \eqref{A:zero}, \eqref{A:smooth}, and \eqref{A:infty}, which will also be important in the sequel, is the \emph{gamma kernel}.
\begin{example}[Gamma kernel]\label{ex:gamma}
Let $g$ be the gamma kernel
\begin{align*}
g(x) = x^{\alpha}e^{-\lambda x}, \qquad x>0, \qquad \alpha \in \left(-\frac{1}{2},\frac{1}{2}\right), \quad \lambda \in \left(0,\infty\right),
\end{align*}
which satisfies \eqref{A:zero}, \eqref{A:smooth}, and \eqref{A:infty}, as shown in \citet[Example 2.1]{BLP15}. With this kernel function, the process $X$ has roughness index $\alpha$ and memory properties controlled by $\lambda$, as per Proposition \ref{prop:BSSsm}. In the following, we will call the $\BSS$ process with the gamma kernel the \emph{Gamma-$\BSS$} process.
\end{example}

We will also need the correlation structure of the Gamma-$\BSS$ process. We easily find
\begin{align*}
c_X(0) := Var(X_t) &=\E[v_0^2] \int_0^{\infty} x^{2\alpha}e^{-2\lambda} dx \\
	&= \E[v_0^2](2\lambda)^{-2\alpha-1} \Gamma (2\alpha + 1),
\end{align*}
where $\Gamma(a) = \int_0^{\infty} x^{a-1}e^{-x}dx$ is the gamma function. For general $h \in \R$, we have the autocovariance function, using \citet[formula 3.383.8]{integrals},
\begin{align*}
c_X(h) := Cov(X_{t+h},X_t) =\E[v_0^2] \frac{\Gamma(\alpha + 1)}{\sqrt{\pi}} \left( \frac{|h|}{2\lambda} \right)^{\alpha + 1/2} K_{\alpha + 1/2}(\lambda |h|),
\end{align*}
where $K_{\nu}(x)$ is the modified Bessel function of the third kind with index $\nu$, evaluated at $x$ \citep[see e.g.,][section 8.4]{integrals}, and the autocorrelation function can be computed using the identity $\rho_X(h) = c_X(h)/c_X(0)$. We note that $c_X$ is the \emph{Mat\'ern covariance function} \citep{matern60,matern93}, which is widely used in many areas of statistics, e.g., spatial statistics, geostatistics, and machine learning. 

\begin{remark}\label{rem:LMvsNLM}
These two examples of kernel functions exemplify the theoretical distinction between long and short memory. In particular, by Proposition \ref{prop:BSSsm}, the Gamma-$\BSS$ process adheres to
\begin{align*}
\rho_X(h) \sim c \cdot e^{-\lambda h} h^{\alpha}, \quad h \rightarrow \infty,
\end{align*}
i.e., it has short memory, while the Power-$\BSS$ process will have polynomially decaying ACF and, in particular, the long memory property when $\gamma < 1$, cf.\ Remark \ref{rem:polyBSS}. Although, theoretically, the Gamma-$\BSS$ process has short memory, by selecting very small values of $\lambda$, it is possible to specify processes with a very high degree of persistence, mimicking long memory on finite time intervals. Empirically, these two $\BSS$ models allow us to assess if there is any gain from using a model with bona fide long memory, as opposed to a highly persistent model with (\mikkolll{formally}) short memory, in particular in terms of forecasting accuracy.
\end{remark}

\mikkolll{
\begin{remark}
Thanks to the moving average representation derived by \cite{BNBO2011}, the fractional Ornstein-Uhlenbeck process is a special case of the $\BSS$ process \eqref{eq:BSS}, whereby the FSV and RFSV models are special cases of the general model formulated in this section. That said, we do not analyze these models for this viewpoint in the current paper due to their limitations discussed earlier.
\end{remark}}

\subsection{Implications for  volatility}
The following results show that when log volatility is a $\BSS$ process, the roughness and memory properties will carry over to the volatility process itself. A result related to Theorem \ref{th:sig}(i) below was given in \cite{CR98} in the case where $X$ is an fBm. Our results here are stated for a Gaussian $\BSS$ process, i.e., when $v_t = v > 0$ is constant for all $t$, but we conjecture that the results hold also for more general $\BSS$ processes, under suitable assumptions.

Let $\sigma = (\sigma_t)_{t \geq 0}$ be as in \eqref{eq:sig} and
\begin{align*}
\rho(h) = \text{Corr}(\sigma_{t+h},\sigma_t), \quad h \in \R.
\end{align*}
The first part of the following theorem shows that $\sigma$ inherits the roughness properties of the $\BSS$ process $X$. The second part shows that the same is true of the long-term memory properties.
\begin{theorem}\label{th:sig}
Let $\sigma$ be given by \eqref{eq:sig} where $X$ is a $\BSS$ process satisfying \eqref{A:zero}, \eqref{A:smooth}, and \eqref{A:infty}, with $v_t = v>0$ for all $t$. Then,
\begin{enumerate}[label=(\roman*),ref=\roman*,leftmargin=*]
\item\label{item:rough} as $ |h| \rightarrow 0$,
\begin{align*}
1- \rho(h) \sim c |h|^{2\alpha+1}L_0(|h|).
\end{align*}
\item\label{item:memory} as $|h| \rightarrow \infty$,
\begin{align*}
\rho(h) \sim c \cdot \rho_X(h).
\end{align*}
\end{enumerate}
\end{theorem}

\subsection{Simulation of the stochastic volatility model}
Fast and efficient simulation of a stochastic volatility model is advantageous for a number of reasons. For instance, one might wish to conduct simulation experiments to assess the properties of the model, or one might wish to price derivatives by Monte Carlo simulation. We discuss here briefly how our model can be simulated rather easily and efficiently. \mikkolll{Numerically} efficient simulation methods for rough volatility models, such as the ones considered in this paper, should not be taken for granted, however. Rough volatility models are typically non-Markovian, depending on the entire history of the process, which makes conventional recursive simulation methods inapplicable. What is more, the possibility of non-Gaussianity of the $\BSS$ process poses further problems, as this rules out simulation methods based on Gaussianity, such as Cholesky factorization and circulant embedding methods \citep[e.g.,][Chapter XI]{AG07}.

The model to be simulated is
\begin{align*}
S_t &= S_0\exp \left( \int_0^t \sigma_s dB_s - \frac{1}{2}\int_0^t \sigma_s^2 ds \right), \\
\sigma_t &= \xi \exp \mikkolll{\left(X_t \right)},
\end{align*}
\mikkolll{where $S_0>0$,  $\xi >0$, $B$ is a Brownian motion,} and $X$ is one of our candidate models for log volatility, presented in Section \ref{sec:model}. We can simulate $S$ on a grid using Riemann-sum approximations of the integrals. To this end, we need to first simulate $B$ and $\sigma$ on the same grid, which boils down to simulating $B$ and $X$. As we typically want to make these processes correlated, to capture the leverage effect, it is necessary to simulate $B$ and $X$ jointly.

When $X$ is Gaussian, for instance a Cauchy process or a $\BSS$ process with constant volatility, it can be simulated exactly using, e.g., a Cholesky factorization of the covariance matrix of the observations \citep[][pp.\ 311--314]{AG07}. One can additionally compute the covariance structure of the Gaussian bivariate process $(B,X)$ and simulate $B$ and $X$ jointly, and in this way account for correlation between the two processes. This was the approach taken in \cite{BFG15}. However, as the authors also note, the Cholesky factorization is computationally expensive and can become even infeasible if the number of observations to be simulated is large. Instead, we recommend using the circulant embedding method \citep[][pp.\ 314--316]{AG07} in the Gaussian case or the \emph{hybrid scheme} of \cite{BLP15} in the general case. The hybrid scheme is tailor-made for $\BSS$ processes and its advantages are that (i) simulation is fast and in most cases accurate\footnote{The hybrid scheme requires truncating the integral representation \eqref{eq:BSS} ``near'' infinity. If the kernel function has extremely slow decay, this truncation may lead to some loss of accuracy, with regards to the memory properties of the process.} (although approximate), (ii) it allows for non-Gaussianity of $X$ through volatility (of volatility) modulation, and (iii) inclusion of leverage, i.e., correlation between $X$ and $B$, is straightforward. 

We refer to \cite{BLP15} for an exposition of the hybrid scheme for the $\BSS$ process $X$, defined by \eqref{eq:BSS}, under assumptions \eqref{A:zero}, \eqref{A:smooth}, and \eqref{A:infty}. The authors explain in the paper \citep[Section 3.1]{BLP15} how to incorporate correlation between $v$ and $W$, while the same procedure can be used to introduce correlation between $W$ and $B$ or, indeed, between $W$, $B$, and $v$.

\section{Estimation of the models}\label{sec:estimate}
We propose a two-step estimation procedure \mikkolll{for the models of Section \ref{sec:model}.} First, we estimate the roughness parameter $\alpha$ semiparametrically using the scaling relationship \eqref{eq:assA} and then the remaining parameters are estimated using a parametric method-of-moments approach. 

We consider two different estimators of $\alpha$, details of which are given in Section \ref{sec:estAlpha}. The first estimator is well-known in the literature \citep[e.g.,][]{CH94,GSP12} and relies on a simple OLS regression. The second estimator is novel; it relies on a non-linear least squares (NLLS) regression and is constructed so that it is robust to measurement noise in the observations. This latter feature is important in the context of this paper, where, when we apply the methods to real data, we do not have access to the spot variance process but only to a noisy estimate thereof. Section \ref{sec:mom} briefly reviews the second method-of-moments step of our proposed estimation approach, which can also be conducted in a noise-robust way.

\subsection{Estimating the roughness parameter of log volatility} \label{sec:estAlpha}
The estimation of the roughness parameter $\alpha$ from observations of the log-variance process $X_t = \log \sigma^2_t$ is well-understood. A particularly simple, and much-used, estimator can be constructed as follows. Equation \eqref{eq:assA} implies that the \emph{second order variogram} $\gamma_2(\cdot)$ satisfies
\begin{align*}
\gamma_2(h) := \mathbb{E}[(\log \sigma_{t+h} - \log \sigma_t)^2] \sim c |h|^{2\alpha +1}, \quad |h| \rightarrow 0, 
\end{align*}
for a constant $c>0$. This motivates running the ordinary least squares (OLS) regression
\begin{align}\label{eq:OLSa}
\log \hat{\gamma}_2(k) = a_0 + a_1 \log k + \epsilon_k, \quad k = 1, 2, \ldots, m,
\end{align}
where $m \in  \N$ is a bandwidth parameter, $\epsilon_k$ is an error \mikkolll{term}, and
\begin{align*}
\mikkolll{\hat{\gamma}_2(k) = \frac{1}{n-k} \sum_{i=1}^{n-k} | \log \sigma_{(i+k)\Delta} - \log \sigma_{i\Delta} |^2}
\end{align*}
is the empirical \mikkolll{variogram corresponding to} $\gamma_2(k\Delta)$. The OLS estimat\mikkolll{or} of $\alpha$ is thus $\hat{\alpha} = \hat{\alpha}_{OLS} = (\hat{a}_1 - 1)/2$, where $\hat{a}_1$ is the OLS estimat\mikkolll{or} of $a_1$ from the regression \eqref{eq:OLSa}. We will \mikkolll{henceforth} refer to this estimator of $\alpha$ as the \mikkolll{\emph{OLS estimator}}. Following \cite{bennedsen16a}, we choose a small value for the bandwidth and set $m = 6$.

Unfortunately,  in \mikkolll{reality} we do not have access to the latent volatility process $\log \sigma_t$. Instead, the available data are estimates of log volatility, $\log \hat{\sigma}_t$, see Section \ref{sec:empirical} below for further details. It is reasonable to ask whether it is possible to accurately estimate the roughness index of the latent log volatility process $\log \sigma_t$ using $\log \hat{\sigma}_{k}$. As we will see, \mikkolll{this concern is legitimate as} the OLS estimator of $\alpha$ can be severely biased in the presence of measurement noise, \mikkolll{which may arise} from using $\log \hat{\sigma}_t$ in place of $\log \sigma_t$.  To alleviate this possible bias, we \mikkolll{formulate} an estimator that is robust to additive noise in the observations. To fix ideas, suppose that the observations are related to the latent process through
\begin{align}\label{eq:u}
\log \hat{\sigma}_k = \log \sigma_{k\Delta} + u_k,
\end{align}
where $u_k$ is a noise \mikkolll{term}. If \mikkolll{the noise terms are} iid with zero mean and $u_k$ is independent of $\log \sigma_{k\Delta}$, we have
\begin{align}\label{eq:gam_star}
\gamma^*_2(h) := \mathbb{E}[(\log \hat{\sigma}_{t+h} - \log \hat{\sigma}_t)^2] = 2\sigma_u^2 + \gamma_2(h), 
\end{align}
where $\gamma_2(\cdot)$ is the second-order variogram of the latent process $\log \sigma_t$ and $\sigma_u^2$ is the variance of $u_k$. When $\sigma_u^2 > 0$,  $\log \gamma^*(h)$ is not linear in $\log h$, and the OLS regression \eqref{eq:OLSa} can not be expected to yield a consistent estimate of $\alpha$ \citep[][Proposition 3.3]{bennedsen16a}. Instead, \eqref{eq:gam_star} motivates a non-linear least least squares (NLLS) regression:
\begin{align*}
\hat{\gamma}^*_2(k) = b_0 + b_1 (k\Delta)^{2\alpha+1} + \epsilon_k, \quad k = 1, 2, \ldots, m,
\end{align*}
where $m \in  \N$ is a bandwidth parameter, $\epsilon_k$ is an error \mikkolll{term}, and
\begin{align*}
\hat{\gamma}^*_2(k) = \frac{1}{n-k} \sum_{i=1}^{n-k} | \log \hat{\sigma}_{(i+k)\Delta} - \log \hat{\sigma}_{i\Delta} |^2
\end{align*}
is the empirical \mikkolll{variogram corresponding to} $\gamma^*_2(k\Delta)$. We refer to this estimator as the \mikkolll{\emph{NLLS estimator} in the sequel}. The following result shows that the NLLS estimator is asymptotically consistent. The proof is given in Appendix \ref{app:proofs}. 

\begin{theorem}\label{th:NLLS}
Let $X$ be a \mikkolll{square-integrable} continuous process with stationary increments and variogram
\begin{align*}
\gamma_2(h) = \E[ |X_{t+h} - X_t|^2] = c_0 h^{2\alpha_0+1}, \quad h > 0,
\end{align*}
\mikkolll{with constant $c_0>0$ and roughness parameter $\alpha_0 \in (-1/2,1/2)$}. Fix $\Delta > 0$ and let
\begin{align*}
Z_{i\Delta} = X_{i\Delta} + u_i , \quad  i = 1, \ldots, n,
\end{align*}
be equidistant observations \mikkolll{such that} $u = \{u_i\}_{i=1}^n$ is a zero-mean iid sequence, independent of $X$, with finite $Var(u_1) = \sigma_u^2 \geq 0$. Denote by $\gamma_2^*(h)$ the variogram of $Z$ and \mikkolll{ let $\Theta$ be a} compact subset of $\R_+ \times \R_+ \times (-1/2,1/2)$ such that $(2\sigma_u^2,c_0,\alpha_0)$ is in the interior of $\Theta$\mikkolll{. Finally,} define the NLLS estimat\mikkolll{or} of $\theta := (a,c,\alpha)$ as
\begin{align*}
\hat{\theta} = (\hat{a},\hat{c},\hat{\alpha}) := \arg \inf_{\theta \in \Theta} \sum_{k=1}^m \left( \hat{\gamma}_2^*(k,\Delta) - a - c (k\Delta)^{2\alpha+1} \right)^2,
\end{align*}
with a fixed bandwidth $m \geq 3$, and
\begin{align*}
\hat{\gamma}_2^*(k,\Delta) := \frac{1}{n-k} \sum_{j=1}^{n-k}|Z_{(j+k)\Delta} - Z_{j\Delta}|^2, \quad k = 1, 2,  \ldots, m.
\end{align*}
Then
\begin{align*}
(\hat{a},\hat{c},\hat{\alpha}) \rightarrow (2\sigma_u^2,c_0,\alpha_0) \quad  \textnormal{in probability,} \quad \textnormal{as} \quad n \rightarrow \infty.
\end{align*}
\end{theorem}

\begin{remark}\label{rem:caveat}
The assumption that
\begin{align*}
\gamma_2(h) =  c_0 h^{2\alpha_0+1}, \quad h > 0,
\end{align*}
in Theorem \ref{th:NLLS} is \mikkolll{ somewhat idealized and does not cover} many models of interest. For instance, \mikkolll{in this paper we mostly work with models such that}
\begin{align*}
\gamma_2(h) =  h^{2\alpha_0+1}L(h), \quad h > 0,
\end{align*}
where the function $L$ is \mikkolll{slowly varying at zero, which we recall means} that $\lim_{x\rightarrow 0} \frac{L(tx)}{L(x)} = 1$ for all $t>0$. It is possible that the presence of the function $L$ could bias the NLLS estimates. However, we conjecture that such bias should typically be small, since, by the properties of slowly varying functions \citep[e.g.,][]{BGT}, it is possible to find positive constants $C_1$ and $C_2$ such that for all $\epsilon>0$ and $h$ sufficiently small,
\begin{align}\label{eq:variogram-potter}
C_1  h^{2\alpha_0+ 1+\epsilon} \leq \gamma_2(h) \leq C_2  h^{2\alpha_0+1 - \epsilon},
\end{align}
indicating that any potential bias in the estimate of $\alpha_0$ can be made arbitrarily small by setting $\Delta$ and $m$ small enough.\footnote{\mikkolll{In fact, in the context of the Gamma-$\BSS$ and Power-$\BSS$ models the function $L$ is actually bounded near the origin so \eqref{eq:variogram-potter} holds with $\epsilon=0$ then.}} Developing an asymptotic result for the NLLS estimators under these more general assumptions is beyond the scope of the present paper. However, as we will see below in the simulation studies of Section \ref{sec:sim}, we find that the presence of the function $L$ has only negligible impact on the NLLS estimates of $\alpha_0$ in settings relevant to the one considered in this paper.
\end{remark}

How to choose the bandwidth $m$ optimally for the NLLS estimator is an open problem, which is beyond the scope of this paper. Instead, in Section \ref{sec:sim1}, we use simulations to guide the choice of bandwidth in practice.

\subsection{Parametric estimation of the remaining parameters by \mikkolll{the method of moments}}\label{sec:mom}
Let $X$ be a \mikkolll{zero-mean stochastic process with unit variance,} satisfying assumption\mikkolll{s} (\ref{A:zero})--(\ref{A:infty}). Consider the model of log-volatility
\begin{align*}
\log \sigma_t = \mu + \nu X_t,
\end{align*}
for constants $\mikkolll{\mu =  \log \xi} \in \R$ and $\nu > 0$. Clearly, $\mu$ and $\nu$ can be estimated by the sample mean and sample standard deviation of $\log \sigma_t$, respectively, while the roughness parameter $\alpha$ of $X$ can be estimated using either the OLS or the NLLS estimator \mikkolll{discussed} above, applied to $\log \sigma_t$. 

The remaining parameters of \mikkolll{the} model \mikkolll{in the vector $\theta$} depend on the particular choice of parametric \mikkolll{specification of} $X$. \mikkolll{For instance, if $X$ is the Cauchy process introduced in Section \ref{sec:cauchy} then $\theta = \beta$, while if $X$ is the Power-$\BSS$ process of Example \ref{ex:power} then $\theta = \gamma$, and if $X$ is the Gamma-$\BSS$ process of Example \ref{ex:gamma} then $\theta = \lambda$.} We propose to estimate $\theta$ by fitting the empirically estimated autocorrelation function of $\log \sigma_t$ to the autocorrelation implied by the parametric model for $X$, given the first-step estimate of $\alpha$. Let the empirical autocorrelation of the data be denoted by $\hat{\rho}(h)$ and the parametric autocorrelation function implied by the model for $X$ be denoted by $\rho(h;\alpha,\theta)$. We can now define the estimator of $\theta$ as 
\begin{align*}
\hat{\theta} = \arg \min_{\theta} \sum_{i=1}^H \left(\hat\rho (h) - \rho(h;\hat \alpha;\theta)\right)^2,
\end{align*}
for some $H \geq 2$, where $\hat\alpha$ is the first-step estimat\mikkolll{or} of $\alpha$, obtained by either the OLS or NLLS \mikkolll{approach} \mikkolll{discussed in Section \ref{sec:estAlpha}}.

A noise-robust alternative to the estimator $\hat \theta$ can be constructed \mikkolll{in a straightforward manner}. Indeed, according to the model for the noisy observations \eqref{eq:u}, we can write
\begin{align*}
\rho^*(h):= Corr(\log \hat \sigma_{k+h}, \log \hat \sigma_k) &= \frac{Cov(\log \sigma_{(k+h)\Delta} + u_{k+h}, \log \sigma_{k\Delta} + u_k)}{Var(\log \sigma_{k\Delta} + u_k)} \\
	&= \frac{Cov(\log \sigma_{(k+h)\Delta}, \log \sigma_{k\Delta} )}{Var(\log \sigma_{k\Delta} ) + Var(u_k)} \\
	& = C\cdot \rho(h), 
\end{align*}
where $\rho(h) := Corr(\sigma_{t+\Delta h},\sigma_t)$, $C = (1+\gamma_{\mathrm{noise}}^2)^{-1}$ is a constant, and the noise-to-signal ratio $\gamma_{\mathrm{noise}}$ is given by
\begin{align*}
\gamma_{\mathrm{noise}}^2 = \frac{Var(u_k)}{Var(\log \sigma_t)}.
\end{align*}
Therefore, the logarithm of the correlation function of the noisy observations \mikkolll{$\log \hat \sigma_i$, $i=1,\ldots,n$,} obeys
\begin{align*}
\log \rho^*(h) = c + \log \rho (h),
\end{align*}
where $c = \log C$ is a constant. In conclusion, we can estimate $\theta$ in \mikkolll{a noise-robust} manner by $\hat \theta^*$ from the non-linear least squares regression
\begin{align*}
(\hat c, \hat{\theta}^*) = \arg \min_{(c,\theta)} \sum_{i=1}^H \left( \log \hat\rho (h) - c - \log \rho(h;\hat \alpha;\theta)\right)^2,
\end{align*}
where $\hat\alpha$ is the \mikkolll{first-step noise-robust NLLS estimator of $\alpha$, developed in Section \ref{sec:estAlpha}.}

%

\section{Simulation studies}\label{sec:sim}
The aim of this section is to investigate the finite sample properties of the estimation procedure proposed in the previous section. The Web Appendix contains \mikkolll{more extensive} simulation results.

Section \ref{sec:sim1} presents results from a small simulation study, which \mikkolll{facilitates} the choice of bandwidth parameter \mikkolll{$m$} for the NLLS estimator of $\alpha$. A brief comparison with the OLS estimator is included. Section \ref{sec:sim2} contains a more in-depth comparison of the performance of the OLS and NLLS estimators when applied to more realistic data. Section  \ref{sec:sim_lam} briefly investigates the properties of the \mikkolll{method-of-moments} estimator of the remaining parameters, presented in Section \ref{sec:mom}.

\subsection{The influence of the bandwidth parameter $m$ on the NLLS estimator}\label{sec:sim1}
To assess the effect of the bandwidth parameter $m$ for the NLLS estimator, we conduct a small simulation study. Consider the model
\begin{align*}
Y_t = \mu + B_t^{\alpha+1/2} + \kappa \epsilon_t,\quad \mikkolll{t  \geq 0,}
\end{align*}
where $\mikkolll{\mu = \log \xi} \in \R$, $\kappa \geq 0$, and \mikkolll{$B^{\alpha+1/2}$} is a fractional Brownian motion with Hurst index $H = \alpha + 1/2$, with $\alpha \in (-1/2,1/2)$, independent of the iid \mikkolll{error terms $\epsilon_t \sim N(0,1)$, $t\geq 0$}. Thus, $Y_t$ satisfies the assumptions underlying Theorem \ref{th:NLLS}. We set $\mu = 0$, $\alpha \in \{-0.4,-0.2,0,0.2\}$ and simulate $n$ observations of $Y_{i\Delta}$ with $\Delta = 0.10$ for $i = 1,2, \ldots, n$. We \mikkolll{generate} $500$ Monte Carlo replications of $Y_t$ and use the NLLS estimator to estimate $\alpha$ for various values of the bandwidth parameter $m$, the true roughness index $\alpha$, and for the noise variance $\kappa$. From these $500$ simulated instances, we record the ``optimal'' value of the bandwidth parameter \mikkolll{$m^*$}, i.e., the value of $m$ that \mikkolll{minimizes the} mean squared error of the estimates of $\alpha$ over the $500$ simulations. The results for $n \in \{1000,2000\}$,  $m\in \{3,4,\ldots, 50\}$, and $\kappa \in \{0,0.1,0.2,0.3,0.4,0.5\}$ are shown in Figure \ref{fig:mopt}.

It is clear that, as the noise level $\kappa$ increases, so does the optimal bandwidth. Likewise, the rougher the underlying process $B_t^{\alpha+1/2}$ is, i.e., the lower $\alpha$ is, the larger the optimal bandwidth $m^*$. Going from $n=1\ 000$ to $n= 2\ 000$ does not seem to have a large effect on $m^*$. For \mikkolll{relevant} values of $\alpha$, i.e.\mikkolll{,} $\alpha \leq 0$, the optimal bandwidth seems to be roughly between $m^* = 10$ and $m^* = 20$, although for values of $\alpha$ very close to $-0.50$, $m^*$ might be even larger. 

The NLLS estimator can be quite \mikkolll{sensitive to} the numerical optimization procedure \mikkolll{employed}. This, together with the uncertainty \mikkolll{over} the precise value of $m^*$ in practical applications, prompts us to suggest to average the estimates coming from several values of the bandwidth parameter $m$. (Alternatively, one could take the median of the estimates instead of the mean.) Given the \mikkolll{admittedly somewhat rudimentary} simulation study of this section, we recommend averaging over the $11$ estimates obtained from setting $m = 10, 11, \ldots, 20$. We have found that this strategy works quite well in practice, but we stress again that a more \mikkolll{systematic} way \mikkolll{of choosing} $m$ would be preferable. This, however, is beyond the scope of the present article and is left for future work.

To illustrate the performance of this strategy, and compare with the OLS estimator, Figure \ref{fig:bias0} plots the bias of the proposed NLLS estimator, i.e.\mikkolll{,} the average of the NLLS estimates from $m = 10, \ldots, 20$, as well as the bias of the OLS estimator, in the same simulation study used to construct Figure \ref{fig:mopt}. It is clear that that the OLS estimator is very biased for $\kappa>0$, while the NLLS estimator seems to alleviate the bias to a \mikkolll{great extent}. The next section further examines the properties of the two estimators in a more realistic setting.

\begin{figure}[!t] 
\centering 
\includegraphics[scale=0.92]{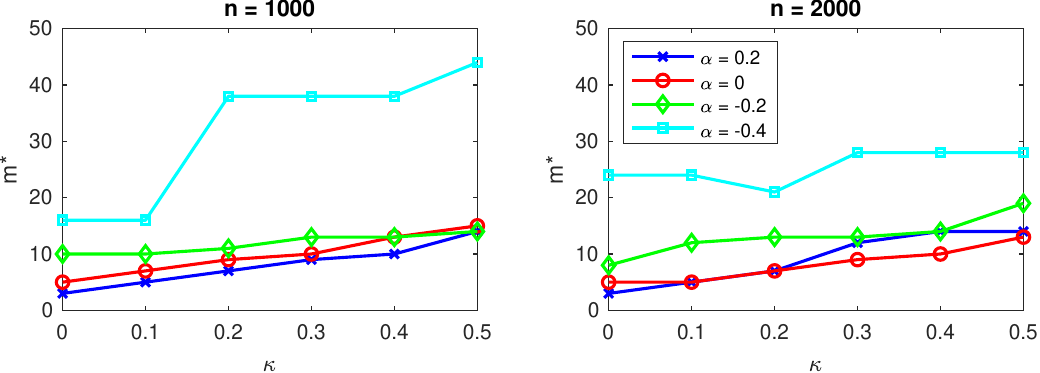}
\caption{\it Determining the optimal (\mikkolll{with respect to mean squared error}) bandwidth $m^*$  for the NLLS estimator using simulations. The simulation setup is described in Section \ref{sec:sim1}.}
\label{fig:mopt}
\end{figure}

\begin{figure}[!t] 
\centering 
\includegraphics[scale=0.92]{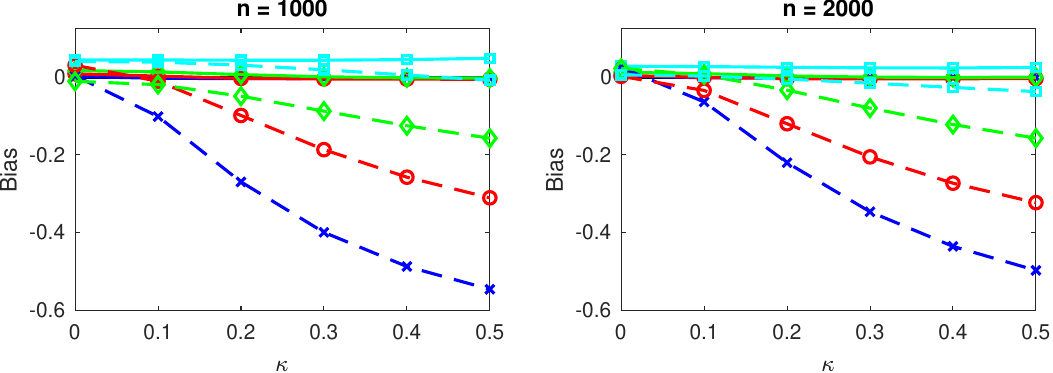}
\caption{\it Biases of the NLLS (\mikkolll{solid} lines) and OLS (dashed lines) estimators in the simulation study of Section \ref{sec:sim1}. The NLLS estimator is the average of \mikkolll{individual NLLS estimators with} bandwidths $m = 10,11, \ldots, 20$, while the OLS estimator uses $m = 6$.}
\label{fig:bias0}
\end{figure}

\subsection{Estimating the roughness index of a noisy time series}\label{sec:sim2}
Consider now the model 
\begin{align*}
Y_t = \mu + \nu X_t + \kappa \epsilon_t, \quad \mikkolll{t \geq 0,}
\end{align*}
where $\mikkolll{\mu = \log \xi}\in \R$, $\nu,\,\kappa>0$, \mikkolll{$X$ is a zero-mean stochastic process with unit variance} satisfying assumption\mikkolll{s} (\ref{A:zero})--(\ref{A:infty}), and $\epsilon_t \sim N(0,1)$\mikkolll{, $t \geq 0$,} are iid error terms. The term $u_t = \kappa \epsilon_t$ represents the measurement noise perturbing the observations of the latent process $Y_t^* = \mu + \nu X_t$. The relevant noise-to-signal ratio is denoted by $s  = \sqrt{Var(u_t)/Var(Y_t^*)} = \kappa/\nu$.

Using this setup, we want to \mikkolll{assess} the finite sample properties of the two estimators of $\alpha$ when $X$ is the Gamma-$\BSS$ process of Example \ref{ex:gamma}, \mikkolll{sampled at daily or higher frequency}. Note that, strictly speaking, this process does not \mikkolll{satisfy} the assumptions behind Theorem \ref{th:NLLS}, cf.\ Remark \ref{rem:caveat}. We set $\lambda = 0.02$, $\mu = 1$, $\nu = 1$ and simulate $n = 1\ 000$ observations with sampling frequency $\Delta>0$, mimicking $n$ observations of log volatility data, observed $1/\Delta$ times per day. We vary the noise-to-signal ratio $s \in \{0,0.1,0.2,0.3,0.4,0.5\}$ by changing the \mikkolll{standard deviation of the noise} $\kappa  \in \{0,0.1,0.2,0.3,0.4,0.5\}$. Two values of $\alpha$ are considered, namely $\alpha = 0$ and $\alpha = -0.35$, which are particularly interesting for the purposes of this paper: $\alpha=0$ corresponds to the ``standard'' semimartingale framework, while $\alpha = -0.35$ is a value often found when estimating the roughness index of realized volatility measures, see, e.g., \cite{GJR14} and Section \ref{sec:empirical} below. Two values are also considered for the sampling frequency, namely $\Delta = 1$ and $\Delta = 0.1$. The first value corresponds to daily observations, while the second value corresponds to \mikkolll{intraday} observations (roughly hourly) of realized log volatility. The value for the mean reversion parameter, $\lambda = 0.02$, corresponds to what we estimate on the S\&P $500$ data in Section \ref{sec:SP} below. We performed similar simulations for different values of $\alpha$, $\lambda$, and $\Delta$, and found very similar results\mikkolll{, which} are presented in the Web Appendix.

We simulate $500$ instances of $Y_t$ and estimate the roughness index $\alpha$ using both the OLS and the NLLS estimator. Table \ref{tab:estRes} reports the mean, median, and standard deviation (SD) of the $500$ estimates of $\alpha$ coming from each of the two estimators. As remarked above, we consider $\alpha = 0$ (Panels A and C) and $\alpha = -0.35$ (Panels B and D), as well as $\Delta = 1$ (Panels A and B) and $\Delta = 0.1$ (Panels C and D). Several things can be learned from the results. \mikkolll{Firstly}, the NLLS estimator is more \mikkolll{dispersed} than the OLS estimator, as evidenced by the larger standard deviation of the estimates coming from the NLLS estimator. \mikkolll{Secondly}, the OLS estimator quickly becomes biased when the \mikkolll{noise-to-signal} ratio $s$ increases. The NLLS estimator, however, is able to alleviate this bias to a \mikkolll{great} exten\mikkolll{t}. This is especially clear in the case of $\alpha = 0$, where the OLS estimator is very biased, while the NLLS estimator is much more accurate. When sampling infrequently, i.e.\mikkolll{,} when $\Delta = 1$, the NLLS estimator is  \mikkolll{slightly downwards biased}, however, due to the mean reversion in the process, cf. Panel A of Table \ref{tab:estRes}. This is evidence of the fact, that when the sampling frequency is low, mean reversion can  \mikkolll{resemble roughness}.  At the mean reversion levels relevant to this paper, i.e.\mikkolll{,} for very persistent processes, the bias is small, however.\footnote{\cite{JE2019} show that rough models can be approximated by a mixture of standard models with varying degrees of mean reversion, including components with arbitrarily fast mean reversion. However, such models are arguably less parsimonious and flexible than the rough models proposed in this paper.} We refer the interested reader to the Web Appendix, where this behavior is studied in more detail. 

\begin{table}
\caption{\it Monte Carlo simulation results} \vspace{.2cm} 
\scriptsize
\begin{center}
\begin{tabularx}{0.98\textwidth}{@{\extracolsep{\stretch{1}}}llcccccc@{}} 
\toprule
\multicolumn{8}{l}{Panel A: $\alpha = 0$, $\Delta = 1$}  \\  
 & $s = $ & $0$ & $0.1$ & $0.2$ & $0.3$ & $0.4$ & $0.5$ \\ 
\midrule
OLS & Mean & $  -0.01 $ & $  -0.10 $ & $  -0.24 $ & $  -0.33 $ & $  -0.38 $ & $  -0.42 $\\
 & Median & $  -0.01 $ & $  -0.10 $ & $  -0.24 $ & $  -0.33 $ & $  -0.38 $ & $  -0.42 $\\
 & SD & $   0.02 $ & $   0.02 $ & $   0.02 $ & $   0.02 $ & $   0.02 $ & $   0.01 $\\
\cmidrule{3-8}    
NLLS & Mean & $  -0.04 $ & $  -0.07 $ & $  -0.08 $ & $  -0.07 $ & $  -0.07 $ & $  -0.06 $\\
 & Median & $  -0.04 $ & $  -0.07 $ & $  -0.08 $ & $  -0.07 $ & $  -0.07 $ & $  -0.07 $\\
 & SD & $   0.04 $ & $   0.04 $ & $   0.06 $ & $   0.07 $ & $   0.08 $ & $   0.09 $\\
\end{tabularx}
\begin{tabularx}{0.98\textwidth}{@{\extracolsep{\stretch{1}}}llcccccc@{}} 
\toprule
\multicolumn{8}{l}{Panel B: $\alpha = -0.35$, $\Delta = 1$}  \\  
 & $s = $ & $0$ & $0.1$ & $0.2$ & $0.3$ & $0.4$ & $0.5$ \\ 
\midrule
OLS & Mean & $  -0.35 $ & $  -0.35 $ & $  -0.37 $ & $  -0.38 $ & $  -0.40 $ & $  -0.41 $\\
 & Median & $  -0.35 $ & $  -0.35 $ & $  -0.37 $ & $  -0.38 $ & $  -0.40 $ & $  -0.41 $\\
 & SD & $   0.02 $ & $   0.02 $ & $   0.02 $ & $   0.02 $ & $   0.02 $ & $   0.02 $\\
\cmidrule{3-8}    
NLLS & Mean & $  -0.33 $ & $  -0.33 $ & $  -0.33 $ & $  -0.33 $ & $  -0.33 $ & $  -0.33 $\\
 & Median & $  -0.34 $ & $  -0.34 $ & $  -0.35 $ & $  -0.35 $ & $  -0.36 $ & $  -0.36 $\\
 & SD & $   0.05 $ & $   0.05 $ & $   0.05 $ & $   0.06 $ & $   0.07 $ & $   0.08 $\\
\end{tabularx}
\begin{tabularx}{0.98\textwidth}{@{\extracolsep{\stretch{1}}}llcccccc@{}} 
\toprule
\multicolumn{8}{l}{Panel C: $\alpha = 0$, $\Delta = 0.1$}  \\  
 & $s = $ & $0$ & $0.1$ & $0.2$ & $0.3$ & $0.4$ & $0.5$ \\ 
\midrule
OLS & Mean & $  -0.00 $ & $  -0.33 $ & $  -0.44 $ & $  -0.47 $ & $  -0.48 $ & $  -0.49 $\\
 & Median & $  -0.00 $ & $  -0.33 $ & $  -0.44 $ & $  -0.47 $ & $  -0.48 $ & $  -0.49 $\\
 & SD & $   0.02 $ & $   0.02 $ & $   0.01 $ & $   0.01 $ & $   0.01 $ & $   0.01 $\\
\cmidrule{3-8}    
NLLS & Mean & $   0.00 $ & $  -0.01 $ & $  -0.01 $ & $   0.01 $ & $   0.02 $ & $   0.02 $\\
 & Median & $   0.00 $ & $  -0.01 $ & $  -0.01 $ & $  -0.00 $ & $   0.02 $ & $   0.05 $\\
 & SD & $   0.05 $ & $   0.07 $ & $   0.11 $ & $   0.18 $ & $   0.24 $ & $   0.28 $\\
\end{tabularx}
\begin{tabularx}{0.98\textwidth}{@{\extracolsep{\stretch{1}}}llcccccc@{}} 
\toprule
\multicolumn{8}{l}{Panel D: $\alpha = -0.35$, $\Delta = 0.1$}  \\  
 & $s = $ & $0$ & $0.1$ & $0.2$ & $0.3$ & $0.4$ & $0.5$ \\ 
\midrule
OLS & Mean & $  -0.35 $ & $  -0.36 $ & $  -0.38 $ & $  -0.40 $ & $  -0.42 $ & $  -0.44 $\\
 & Median & $  -0.35 $ & $  -0.36 $ & $  -0.38 $ & $  -0.40 $ & $  -0.42 $ & $  -0.43 $\\
 & SD & $   0.02 $ & $   0.02 $ & $   0.02 $ & $   0.02 $ & $   0.01 $ & $   0.01 $\\
\cmidrule{3-8}    
NLLS & Mean & $  -0.32 $ & $  -0.32 $ & $  -0.32 $ & $  -0.33 $ & $  -0.32 $ & $  -0.32 $\\
 & Median & $  -0.33 $ & $  -0.34 $ & $  -0.34 $ & $  -0.35 $ & $  -0.34 $ & $  -0.34 $\\
 & SD & $   0.05 $ & $   0.06 $ & $   0.06 $ & $   0.07 $ & $   0.09 $ & $   0.10 $\\
\bottomrule
\end{tabularx}
\end{center}
{\footnotesize \it Monte Carlo simulation study of the OLS and NLLS estimators of $\alpha$. DGP: Gamma-$\BSS$ with $\lambda = 0.02$. Also $\mu = 1, \nu = 1$, $n = 1\ 000$, $T = n\Delta$. Number of Monte Carlo simulations: $500$.} 
\label{tab:estRes}
\end{table}

\subsection{Estimating the remaining parameters \mikkolll{using the method of moments}}\label{sec:sim_lam}
Table \ref{tab:estRes_lam} presents the estimation results using the moment-based estimator of $\lambda$ given in Section \ref{sec:mom}. For brevity, we only report the results for $\Delta=1$\mikkolll{, as} the results for $\Delta = 0.1$ are similar and can be found in the Web Appendix. The \mikkolll{simulation procedure is identical to the one} used in Section \ref{sec:sim2} to study the estimators of $\alpha$, i.e.\mikkolll{, $X$ is a \mikkolll{Gamma-}$\BSS$ process} with $\lambda = 0.02$. \mikkolll{Panels A and B of Table \ref{tab:estRes_lam} contain the results for $\alpha = 0$ and $\alpha = -0.35$, respectively.} Mean, median and standard deviation are \mikkolll{reported using both the ordinary estimator $\hat \theta = \hat \lambda$ and the robust} estimator $\hat \theta^* = \hat \lambda^*$, cf. Section \ref{sec:mom}.

\begin{table}
\caption{\it Monte Carlo simulation results} \vspace{.2cm} 
\scriptsize
\begin{center}
\begin{tabularx}{0.98\textwidth}{@{\extracolsep{\stretch{1}}}llcccccc@{}} 
\toprule
\multicolumn{8}{l}{Panel A: $\alpha = 0$, $\Delta = 1$}  \\  
 & $s = $ & $0$ & $0.1$ & $0.2$ & $0.3$ & $0.4$ & $0.5$ \\ 
\midrule
Standard & Mean & $   0.02 $ & $   0.02 $ & $   0.02 $ & $   0.02 $ & $   0.02 $ & $   0.02 $\\
 & Median & $   0.02 $ & $   0.02 $ & $   0.02 $ & $   0.02 $ & $   0.02 $ & $   0.02 $\\
 & SD & $   0.01 $ & $   0.01 $ & $   0.01 $ & $   0.01 $ & $   0.01 $ & $   0.01 $\\
\cmidrule{3-8}    
Robust & Mean & $   0.02 $ & $   0.02 $ & $   0.02 $ & $   0.02 $ & $   0.02 $ & $   0.02 $\\
 & Median & $   0.02 $ & $   0.02 $ & $   0.02 $ & $   0.02 $ & $   0.02 $ & $   0.02 $\\
 & SD & $   0.01 $ & $   0.01 $ & $   0.01 $ & $   0.01 $ & $   0.01 $ & $   0.01 $\\
\end{tabularx}
\begin{tabularx}{0.98\textwidth}{@{\extracolsep{\stretch{1}}}llcccccc@{}} 
\toprule
\multicolumn{8}{l}{Panel B: $\alpha = -0.35$, $\Delta = 1$}  \\  
 & $s = $ & $0$ & $0.1$ & $0.2$ & $0.3$ & $0.4$ & $0.5$ \\ 
\midrule
Standard & Mean & $   0.03 $ & $   0.03 $ & $   0.03 $ & $   0.03 $ & $   0.03 $ & $   0.03 $\\
 & Median & $   0.03 $ & $   0.03 $ & $   0.03 $ & $   0.03 $ & $   0.02 $ & $   0.02 $\\
 & SD & $   0.02 $ & $   0.02 $ & $   0.02 $ & $   0.02 $ & $   0.02 $ & $   0.03 $\\
\cmidrule{3-8}    
Robust & Mean & $   0.03 $ & $   0.03 $ & $   0.03 $ & $   0.03 $ & $   0.02 $ & $   0.02 $\\
 & Median & $   0.03 $ & $   0.03 $ & $   0.02 $ & $   0.02 $ & $   0.02 $ & $   0.02 $\\
 & SD & $   0.01 $ & $   0.02 $ & $   0.02 $ & $   0.02 $ & $   0.02 $ & $   0.02 $\\
\bottomrule
\end{tabularx}
\end{center}
{\footnotesize \it Monte Carlo simulation study of the estimators of $\lambda$, see Section \ref{sec:mom}. DGP: Gamma-$\BSS$ with $\lambda = 0.02$. Also $\mu = 1, \nu = 1$, $n = 1\ 000$, $T = n\Delta$. Number of Monte Carlo simulations: $500$.} 
\label{tab:estRes_lam}
\end{table}

\section{Application to realized volatility data}\label{sec:empirical}

We \mikkolll{employ a general} model for high-frequency asset returns, \mikkolll{where} the \emph{efficient log price} $Y = (Y_t)_{t \geq 0}$ of \mikkolll{the} asset \mikkolll{follows} an It\^o semimartingale
\begin{align}\label{eq:dS}
dY_t= \mu_tdt + \sigma_t dB_t + dJ_t, \quad t\geq 0,
\end{align}
\mikkolll{specified using standard Brownian motion $B = (B_t)_{t\geq 0}$, drift process $\mu = (\mu_t)_{t\geq 0}$, jump process $J = (J_t)_{t\geq 0}$ and volatility process $\sigma = (\sigma_t)_{t\geq 0}$, satisfying the usual assumptions of adaptedness and local boundedness.} 
Since we sample prices at high frequency, we follow the standard practice of high-frequency financial econometrics by treating the data as noisy observations of the efficient log price $Y$, contaminated by \emph{market microstructure noise},
\begin{align*}
\log S_t =  Y_t + U_t, \quad t \geq 0,
\end{align*}
where $U = (U_t)_{t \geq 0}$ is a microstructure noise process.

Since our focus will be on the volatility process $\sigma$, we do not go into much detail regarding the jump process $J$ or the noise process $U$. We \mikkolll{impose} only mild \mikkolll{standard} assumptions on these processes, such that the jump- and noise-robust estimator of the volatility process, explained below, will be consistent. Sufficient, but not necessary, conditions for this are, e.g., that the jump term is given by a compound Poisson process and the noise term\mikkolll{s are} iid with zero-mean and finite fourth moment, and independent of the efficient price $Y$ \citep[][Proposition 1]{COP14}. Also, in the context of high-frequency returns, the drift process $\mu$ is empirically negligible, and will be ignored from now on. The following sections will explain the data we use for $S$, how we estimate the latent volatility process $\sigma$, and the subsequent empirical findings on this process.

We seek to extract the realized spot volatility process $(\sigma_t)_{t \in [0,T]}$, for some time horizon $T>0$,  from high-frequency observations of the asset price $S$. As $\sigma$ is not directly observable, we need to construct a proxy for it. In particular, in this section, we are interested in assessing \emph{intraday} variation of volatility. This should be contrasted with, e.g., \cite{GJR14}, who consider volatility proxies computed at daily frequency.

To this end, we first specify a step size $\Delta > 0$ such that $T=n\Delta$ for some large $n \in \N$. Then we aim to estimate the \emph{integrated variance} (IV),
\begin{align*}
IV_{t}^{\Delta} := \int_{t-\Delta}^t \sigma^2_sds, \quad t = \Delta, 2\Delta, \ldots, n\Delta.
\end{align*}
Estimators of IV have been extensively studied, prominent examples being realized variance \citep{ABDL01,BNS02a}, realized kernels \citep{BNHLS08}, two-scale estimators \citep{ZMS05}, and pre-averaging methods \citep{JLMPV09}. Except for the first, these methods are robust to market microstructure effects, which is crucial when using prices sampled at higher frequencies \citep[e.g.,][]{PeterAsgerRV06}. Further, as we will see below, the pre-averaging methods \mikkolll{can be adapted in a straightforward manner} to handle jumps in the price process, which is our main reason for choosing this particular approach.

By letting $\Delta$ be sufficiently small, \mikkolll{so that variation in the level of volatility in each time interval of size $\Delta$ is small,} we can use the proxy
\begin{align}\label{eq:spotvol}
\hat{\sigma}^2_t = \Delta^{-1} \widehat{IV}_t^{\Delta}, \quad t = \Delta, 2\Delta,\ldots, n\Delta,
\end{align}
where $\widehat{IV}_t^{\Delta}$ is an estimate of IV derived from one of the methods mentioned. A priori, we do not want to fix any specific value of the step size $\Delta$ to sample volatility, as the choice of this parameter would vary from application to another, and as there is no canonical choice when $\Delta$ is less than a day.  For this reason, we perform the subsequent analyses for various values of $\Delta$ and, \mikkolll{as we shall see}, our empirical findings are similar for essentially all intraday, as well as daily, time scales.

The proxy \eqref{eq:spotvol} can be seen as a (finite difference) time derivative of the estimate of \mikkolll{IV}. We assume that these data \mikkolll{can be described by} a model of the kind considered in the simulation study above, i.e.\mikkolll{,}
\begin{align*}
\log \hat{\sigma}_{k\Delta} = \mu + \nu X_{k\Delta} + \kappa \epsilon_k, \quad k = 1, 2, \ldots, n,
\end{align*}
where $\epsilon_k \sim N(0,1)$ is \mikkolll{a term that captures the measurement error arising from the use of} $\log \hat \sigma_k$, as given in \eqref{eq:spotvol}, in place of the true log volatility process $\log \sigma_{k\Delta}=  \mu + \nu X_{k\Delta}$. Related estimators of spot volatility have been suggested in the literature; see, e.g., \cite{kristensen10}, \cite{BJK12}, and \cite{ZB14}. In this paper we will restrict attention to \eqref{eq:spotvol} where we estimate IV \miksim{using the so-called \emph{pre-averaged bipower variation} measure, developed in \cite{JLMPV09}. In what follows, this measure will be denoted by $BV^{\Delta,*}_t$; \mikkolll{its} implementation is briefly reviewed in Appendix \ref{app:BV}. The statistic $BV^{\Delta*}$ is robust to both jumps ($J$) and market microstructure noise ($U$), which is why we choose to use this particular measure of IV. We \mikkolll{have} also conducted analogous analyses using the so-called pre-averaged realized variance $RV^{\Delta*}$ (see Appendix \ref{app:BV} for the definition)\mikkolll{, obtaining results that consistent with those in the present paper.}}\footnote{We \mikkolll{have} also \mikkolll{carried out} the analyses using the realized kernel estimator of \cite{BNHLS08}, obtaining similar results. The details can be found in the first version of this paper, available at:\ \url{https://arxiv.org/1610.00332v1}.}

Lastly, we ask what a realistic signal-to-noise ratio \mikkolll{$s$} is in the data studied here. The answer to this question is important, since it will allow us to judge whether we can rely on the promising \mikkolll{simulation} results of Section \ref{sec:sim} and, with reasonable confidence, \mikkolll{rely upon} the estimates coming from our \mikkolll{various} estimators. To \mikkolll{work out} a rough estimate of \mikkolll{$s$}, note that the related spot volatility estimator of \cite{kristensen10} adheres to \citep[][Theorem 3]{kristensen10}
\begin{align*}
\hat \sigma^2_t = \sigma_t^2 + \eta_t , \quad \eta_t \sim MN(0,2 K \sigma^4_t /n),
\end{align*}
where  \mikkolll{$MN$} denotes a mixed normal distribution, indicating that $\eta_t$ is normally distributed, conditional on $\sigma_t$, while $n$ is the number of observations used to compute the estimate $\hat \sigma^2_t$. The constant $K$ is related to the kernel function used in the estimator of  \cite{kristensen10}. For instance, $K = 0.6$ for the so-called Epanechnikov kernel and $K = 0.5$ for the uniform kernel. This implies that, in this setup,
\begin{align*}
s = \sqrt{2K/(n\Delta)}.
\end{align*}

Assuming $K \approx 1$, then, using these rough calculations, for daily IV ($\Delta = 1$ day) with prices sampled every $5$ minutes ($n = 6.5\cdot 60 /5 = 78$) we get $s = \sqrt{2 /(78\cdot1)} = 0.16$, while for hourly IV ($\Delta = 65$ minutes) with prices sampled every \mikkolll{minute} we have $s = \sqrt{2 /(65/6.5)} = 0.45$. These calculations ignore the possible presence of measurement noise, which should increase the noise-to-signal ratio. On the other hand, we use tick-by-tick data, so $n$ will be large in our applications. \mikkolll{This suggests} that $s \in [0,0.5]$ \mikkolll{is a plausible range of values} for the noise-to-signal ratio in our data, at least for $\Delta \geq 30$ minutes, say. (For smaller $\Delta$, say $10$ or $15$ minutes, the noise-to-signal ratio could be even \mikkolll{higher} than $0.5$.) We conclude that the  simulation results of Section \ref{sec:sim} \mikkolll{should} give a reasonable idea of how the estimators will behave when applied to our data. In particular, the OLS estimator \mikkolll{is prone to being} downwards biased, especially if the true value of \mikkolll{the} roughness parameter $\alpha$ is close to $0$, while the NLLS estimators should be able to alleviate this bias \mikkolll{significantly}.

\subsection{Application to the S\&P500 E-mini futures contract}\label{sec:SP}

We analyze tick-by-tick transaction data on the front month E-mini S\&P $500$ futures contract, traded on the CME Globex electronic trading platform, from January 2, 2011 until December 31, 2014 excluding weekends and holidays. After removing days which were not full trading days we arrive at a total of 996 days in our sample. As we are interested in assessing volatility at intraday time scales, we rely on there being a lot of trading activity on the underlying asset. For this reason we restrict our attention to the period of the day when most trading is taking place; this is when the \mikkolll{cash equity market} is open, from 9.30 a.m.\ until 4 p.m.\ Eastern \mikkolll{ Time (ET).}

It is well-known that intraday volatility displays significant seasonality \citep[e.g.,][]{AB97,AB98}. In particular, the ``U-shape'' is ubiquitous, where volatility is high at the opening and at the close of the market, while being lower around midday. \citep[See, e.g.,][Figure 1.]{ABKO16} When $\Delta < 1$ day, it is therefore important to control for this seasonality before performing any further analyses as subsequent estimates could be affected if one does not take this into account \citep[][]{RF15}. We use a multiplicative decomposition
\begin{align*}
\sigma_t = \sigma_t^s \tilde{\sigma}_t, \quad t \geq 0,
\end{align*}
where $\sigma^s$ is the seasonal component and $\tilde{\sigma}$ is the deseasonalized stochastic process we are interested in. To estimate $\sigma^s$ we use the \emph{flexible Fourier form} (FFF) approach of \cite{AB97,AB98}. We then estimate $\tilde{\sigma}^2_t$ by
\begin{align*}
\widehat{\tilde{\sigma}^2_t} = \Delta^{-1}\widehat{IV}^{\Delta}_t/\widehat{(\sigma_t^s)^2} = \Delta^{-1}BV^{\Delta*}_t/\widehat{(\sigma_t^s)^2}, \quad t= \Delta, 2\Delta, \ldots,n\Delta,
\end{align*}
and will from now on be working with these de-seasonalized data when $\Delta < 1$ day. Further, we abuse notation slightly and will write $\sigma$ even though we actually refer to the de-seasonalized process $\tilde{\sigma}$.

An earlier version of this paper contains an in-depth analysis of the dynamic properties of the \mikkolll{de-seasonalized} E-mini S\&P $500$ log volatility data.\footnote{Available at \url{https://arxiv.org/abs/1610.00332v2}. Detailed results are also available from the authors upon request.} The main findings are that: (i) The log volatility process appears to be \mikkolll{approximately} stationary for all values of $\Delta$; (ii) semiparametric estimates of the long-term decay rate of the autocorrelations of log volatility show evidence of very slowly decaying autocorrelations; and (iii) the empirical probability distribution of log volatility appears to deviate noticeably from the Gaussian distribution, especially for log volatility measured at intraday horizons, i.e. for $\Delta < 1$ day. Indeed, the Normal Inverse Gaussian distribution \citep[e.g.,][]{BN97} provides a good fit for these empirical probability distributions.

Panel A of Table \ref{tab:estBoth} reports the estimates of $\alpha$, $\beta$, and  $\lambda$ for the different models proposed in Section \ref{sec:model}, using the data on log volatility of the E-mini S\&P $500$ futures contracts, extracted as explained above.  We used $H = \lceil n^{1/3}\rceil$ lags in the ACF for the method-of-moments estimation procedure. While this choice \mikkolll{may seem somewhat} arbitrary, the results \mikkolll{are} robust to alternative choices. We \mikkolll{find} that both the OLS and NLLS estimators of $\alpha$ indicate that the realized measures of log-volatility are rough with $\hat \alpha \in [-0.3,-0.40]$. The estimates of $\alpha$ are slightly lower for the \mikkolll{smaller} values of $\Delta$, possibly \mikkolll{due to the idea} that the data here can be expected to \mikkolll{exhibit higher levels of} noise, \mikkolll{in relative terms,} as discussed above. The estimates of the memory parameters are very low in all cases\mikkolll{, Indeed}, when estimating the Cauchy and Power-\mikkolll{$\BSS$} models, \mikkolll{we find} that $\beta < 1$, \mikkolll{which suggests} that the data generating process might possess the long memory property. The exponential decay parameter \mikkolll{$\lambda$} of the Gamma-\mikkolll{$\BSS$} process is also estimated to be a very low value, on the order of $0.01$. Estimates of the memory parameters appear to be insensitive to the value of $\Delta$ and to whether or not a noise-robust estimation procedure is used. (In \mikkolll{Table \ref{tab:estBoth}}, an asterisk denotes a noise-robust estimate.) In conclusion, it appears that the realized measures of the log volatility of the E-mini S\&P $500$ futures data studied in this section are rough and \mikkolll{exhibit} a \mikkolll{high} degree of persistence. 

Although we have sought to alleviate the detrimental effects of measurement \mikkolll{error} by using estimators which are (theoretically) robust to noise, the conclusions of course come with the caveat that the \emph{true} stochastic volatility is unobserved. Hence, strictly speaking, \mikkolll{the above conclusions about the nature of volatility can only be drawn insofar the realized measure $\sigma_t$ is concerned.} Nevertheless, the simulation results of Section \ref{sec:sim2}, see Table \ref{tab:estRes}, seem to indicate that if the true value of the roughness index \mikkolll{were} $\alpha = 0$, the OLS estimator would be severely biased, but the NLLS estimator would not. Hence, even though the noisy nature of the data makes it impossible to pinpoint the exact value of $\alpha$ for the underlying latent volatility process, it is, according to the simulation results, somewhat implausible that the true value of $\alpha$ \mikkolll{would be} $\alpha = 0$.\footnote{\mikkolll{This conclusion, of course, has its intrinsic caveats. That is, \mikkolll{it is} predicated on \mikkolll{certain fundamental} assumptions, such as the efficient log price process being given by \eqref{eq:dS}, the noise-to-signal ratio $s$ not being  excessively large, log volatility not being subject to jumps or structural breaks, etc.}} 

Panels B and C of Table \ref{tab:estBoth} report the results from similar analyses, but now using (non-logarithmic) volatility and variance as the data, respectively.\footnote{\label{foot:PanelC}For the variance time series (Panel C), we excluded the three largest observations from the data before performing the analysis, since these observations were so extreme that the estimates of the parameters, in particular $\hat \alpha_{NLLS}^*$ appeared to be unstable. These three large observations, or “outliers”, all occur in $2011$, i.e. in the first year of the sample. We conducted the analysis with both three observations and the entire year $2011$ removed, obtaining similar results. To keep the results as comparable as possible to those for volatility and log volatility, we report them using the former approach which removes far fewer observations. The results using the latter approach are available upon request.} The estimates of the parameters obtained from the volatility (Panel B) and variance (Panel C) time series are very similar to those of the log-volatility time series (Panel A). Overall, it appears that volatility and variance largely inherit the roughness and memory properties of log volatility, as suggested by Theorem \ref{th:sig}.






\begin{table}
\caption{\it Estimation results: E-mini S\&P500 futures data}
\scriptsize
\vspace{.2cm} 
\begin{center}
\begin{tabularx}{0.995\textwidth}{@{\extracolsep{\stretch{1}}}lcccccccc@{}} 
\toprule
\multicolumn{9}{l}{Panel A: Log-volatility}  \\  
 $\Delta$ & $\hat \alpha_{OLS}$ & $\hat \alpha_{NLLS}^*$ & $\hat \beta_{Cau}$ & $\hat \beta^*_{Cau}$ & $\hat \beta_{pBSS}$ & $\hat \beta^*_{pBSS}$ & $\hat \lambda_{gBSS}$ & $\hat \lambda^*_{gBSS}$ \\ 
\midrule
$10$ minutes & $  -0.38 $ & $  -0.38 $ & $   0.17 $ & $   0.18 $ & $   0.31 $ & $   0.26 $ & $   0.03 $ & $   0.02 $\\
$15$ minutes & $  -0.35 $ & $  -0.37 $ & $   0.18 $ & $   0.18 $ & $   0.02 $ & $   0.26 $ & $   0.04 $ & $   0.03 $\\
$30$ minutes & $  -0.31 $ & $  -0.35 $ & $   0.18 $ & $   0.18 $ & $   0.02 $ & $   0.31 $ & $   0.03 $ & $   0.03 $\\
$65$ minutes& $  -0.30 $ & $  -0.35 $ & $   0.18 $ & $   0.18 $ & $   0.30 $ & $   0.31 $ & $   0.02 $ & $   0.02 $\\
$130$ minutes & $  -0.32 $ & $  -0.35 $ & $   0.00 $ & $   0.17 $ & $   0.29 $ & $   0.30 $ & $   0.02 $ & $   0.02 $\\
$1$ day & $  -0.30 $ & $  -0.33 $ & $   0.00 $ & $   0.18 $ & $   0.32 $ & $   0.34 $ & $   0.02 $ & $   0.02 $\\
\end{tabularx}
\begin{tabularx}{0.995\textwidth}{@{\extracolsep{\stretch{1}}}lcccccccc@{}} 
\toprule
\multicolumn{9}{l}{Panel B: Volatility}  \\  
 $\Delta$ & $\hat \alpha_{OLS}$ & $\hat \alpha_{NLLS}^*$ & $\hat \beta_{Cau}$ & $\hat \beta^*_{Cau}$ & $\hat \beta_{pBSS}$ & $\hat \beta^*_{pBSS}$ & $\hat \lambda_{gBSS}$ & $\hat \lambda^*_{gBSS}$ \\ 
\midrule
$10$ minutes & $  -0.33 $ & $  -0.35 $ & $   0.21 $ & $   0.21 $ & $   0.02 $ & $   0.29 $ & $   0.06 $ & $   0.03 $\\
$15$ minutes & $  -0.32 $ & $  -0.36 $ & $   0.19 $ & $   0.19 $ & $   0.02 $ & $   0.21 $ & $   0.04 $ & $   0.02 $\\
$30$ minutes & $  -0.29 $ & $  -0.38 $ & $   0.00 $ & $   0.15 $ & $   0.23 $ & $   0.23 $ & $   0.02 $ & $   0.01 $\\
$65$ minutes&  $  -0.35 $ & $  -0.36 $ & $   0.00 $ & $   0.16 $ & $   0.26 $ & $   0.27 $ & $   0.02 $ & $   0.02 $\\
$130$ minutes & $  -0.35 $ & $  -0.28 $ & $   0.22 $ & $   0.23 $ & $   0.47 $ & $   0.48 $ & $   0.04 $ & $   0.03 $\\
$1$ day & $  -0.23 $ & $  -0.31 $ & $   0.00 $ & $   0.19 $ & $   0.33 $ & $   0.35 $ & $   0.02 $ & $   0.02 $\\
\end{tabularx}
\begin{tabularx}{0.995\textwidth}{@{\extracolsep{\stretch{1}}}lcccccccc@{}} 
\toprule
\multicolumn{9}{l}{Panel C: Variance}  \\  
 $\Delta$  & $\hat \alpha_{OLS}$ & $\hat \alpha_{NLLS}^*$ & $\hat \beta_{Cau}$ & $\hat \beta^*_{Cau}$ & $\hat \beta_{pBSS}$ & $\hat \beta^*_{pBSS}$ & $\hat \lambda_{gBSS}$ & $\hat \lambda^*_{gBSS}$ \\ 
\midrule
$10$ minutes & $  -0.37 $ & $  -0.38 $ & $   0.26 $ & $   0.27 $ & $   0.73 $ & $   0.43 $ & $   0.11 $ & $   0.06 $\\
$15$ minutes & $  -0.35 $ & $  -0.37 $ & $   0.23 $ & $   0.23 $ & $   0.50 $ & $   0.27 $ & $   0.07 $ & $   0.03 $\\
$30$ minutes & $  -0.33 $ & $  -0.39 $ & $   0.00 $ & $   0.16 $ & $   0.26 $ & $   0.26 $ & $   0.02 $ & $   0.02 $\\
$65$ minutes& $  -0.36 $ & $  -0.36 $ & $   0.20 $ & $   0.21 $ & $   0.38 $ & $   0.38 $ & $   0.03 $ & $   0.03 $\\
$130$ minutes & $  -0.34 $ & $  -0.35 $ & $   0.00 $ & $   0.20 $ & $   0.36 $ & $   0.38 $ & $   0.02 $ & $   0.02 $\\
$1$ day & $  -0.35 $ & $  -0.37 $ & $   0.00 $ & $   0.16 $ & $   0.25 $ & $   0.26 $ & $   0.01 $ & $   0.01 $\\
\bottomrule 
\end{tabularx}
\end{center}
{\footnotesize \it Estimates of $\alpha$ using the OLS and NLLS estimators, as well as \mikko{method-of-moments estimates obtained by} matching the empirical ACF with the theoretical ACF in our three parametric models, the Cauchy, the Power- and Gamma-$\BSS$ models. We used \mikkel{$H = \lceil n^{1/3}\rceil$} lags \mikko{to estimate the} memory parameters $\lambda_{BSS}$, $\beta_{BSS}$, and $\beta_{Cauchy}$. \mikkel{The estimate of $\beta_{BSS}$ is calculated from the estimate of method-of-moment\mikkof{s} estimate of $\gamma$ of the Power-$\BSS$ process; that is, when $\hat{\gamma} > 1$, set $\hat{\beta}_{BSS} = \hat{\gamma}$, otherwise set $\hat{\beta}_{BSS} = 2\hat{\gamma}-1$.} Asterisks denote noise-robust estimates. The volatility series in Panels A--C are extracted from E-mini S\&P500 futures data over the period $2011$--$2014$ using the methods described in the main text (regarding Panel C, see also footnote \ref{foot:PanelC}).}
\label{tab:estBoth}
\end{table}

\subsubsection{Does roughness change over time?}\label{sec:time}
One may wonder whether the degree of roughness of volatility changes over time. This very question was also studied by \cite{GJR14}, who divided their volatility data into two parts and found tentative evidence that volatility was less rough during a period that overlapped with the financial \mikkolll{turmoil} of 2008 and 2011.

Our methodology and data allow us to investigate this question more precisely and systematically. To this end, we still use transaction data on E-mini S\&P $500$, but now over a longer period from January 3, 2005 until December 31, 2014. Figure \ref{fig:aTime} provides results of rolling-window estimation of $\alpha$ using the NLLS estimator when the volatility is proxy is constructed using $\Delta = 65$ minutes. The window length \mikkolll{is} $900$ observations, corresponding to roughly $150$ days. We \mikkolll{have also} experimented with different window sizes and different values of $\Delta$, but the results \mikkolll{are rather consistent with those presented here}. Figure \ref{fig:aTime} additionally displays the overall median over the entire period, as well as a smoothed \mikkolll{time series} of the estimates. The figure shows that $\alpha$ does indeed appear to vary in time. In particular, we observe two peaks of ``smoothness'', in the smoothed estimates, that both seem to coincide with periods of market \mikkolll{distress}. The first large and long period of smoothness \mikkolll{occurs} during the financial crisis of 2007--2008, while the second, somewhat shorter and less extreme, period of smoothness \mikkolll{takes place} during 2011, \mikkolll{at} the nadir of the Greek debt crisis.

While it seems imprudent to draw any definite conclusions from a single time series, the findings presented here, together with the empirical evidence of \cite{GJR14}, seem to indicate that volatility exhibits less roughness during periods of market turmoil, possibly due to more sustained trading \mikkolll{activity} occurring in such times. We investigate this further using \mikkolll{a large panel of} equities in Section \ref{sec:Liq} below.

\begin{figure}[!t] 
\centering 
\includegraphics[scale=0.9]{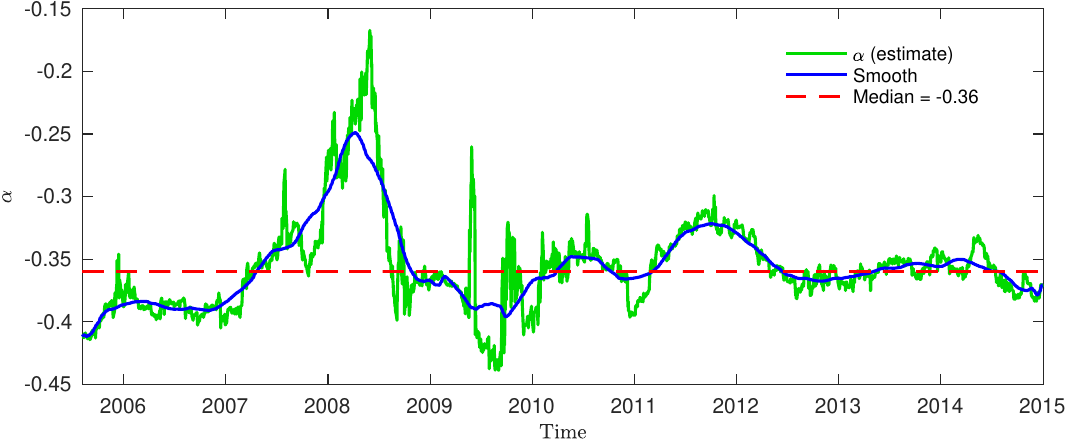} 
\caption{\it Rolling-window NLLS estimates of $\alpha$. The volatility proxy is computed using $\Delta = 65$ minutes and the smoothed version is a simple moving average filter using $75$ observations on both sides of the particular estimate.}
\label{fig:aTime}
\end{figure}

\subsection{Application to individual equities}\label{sec:stocks}
Our goal here is to study whether the findings above are only valid for the E-mini S\&P $500$ futures contract, or whether they hold true more generally. We therefore study volatility data on a large number of US equities. The data consist of daily pre-averaged bipower variation measures (i.e., $BV^{\Delta*}$ with $\Delta = 1$ day), computed using transaction prices obtained from the Trades and Quotes (TAQ) database.

The data set at our disposal runs from January 4, 1993 to December 31, 2013, while the data on some assets \mikkolll{may cover this period only partially}. In total there are $10\ 744$ assets in the sample, classified into ten industry sectors according to the Global Industry Classification Standard.\footnote{See: \url{https://msci.com/gics}} \mikkolll{Since the charactistics of US equities markets have arguably evolved substantially during the past three decades through the decimalization of prices, electronification of trading and proliferation of competing trading venues, we only work with the data over the period
from January 2, 2003 to December 31, 2013, striving for relevance in the present day and beyond.} Additionally, to ensure the reliability of the volatility estimates, we retain only the most liquid assets. Here, a \emph{liquid asset} is characterized by the following \mikkolll{subjective} criteria. 
\begin{enumerate}[label=(\alph*),ref=\alph*,leftmargin=3em]
\item
	The asset is traded on at least $400$ days.
\item
	The maximum number of days when the asset is not traded (on average) every $5$ minutes is $19$.
\item
	The estimate $BV^*$ of IV is strictly positive on every day.\footnote{Although IV is by definition never negative, $BV^*$ can become negative due to the bias correction term appearing in its definition, see Appendix \ref{app:BV}.}
\end{enumerate}
After discarding the assets that do not fulfil these criteria, we are left with 1\,944 liquid assets for our analysis.\footnote{In an earlier version of this paper, available at:\ \url{https://arxiv.org/1610.00332v1}, we \mikkolll{have} employed the realized kernel estimator of \cite{BNHLS08}, which is guaranteed to be non-negative, enabling us to analyze 5\,071 assets. The results therein are very similar to \mikkolll{those in} the present paper.}

The estimates of the roughness parameter $\alpha$ of volatility, using the noise-robust NLLS estimator, for the assets in the sample are summarized in a box plot in Figure \ref{fig:a_box}. A few features of this plot are worth highlighting. Firstly, the estimates of $\alpha$ are practically all negative, mostly in the range $(-0.45,-0.3)$, indicating pronounced roughness. Secondly, there does not seem to be significant \mikkolll{systematic} differences in the roughness estimates across sectors.

\begin{figure}[!t] 
\centering 
\includegraphics[scale=0.9]{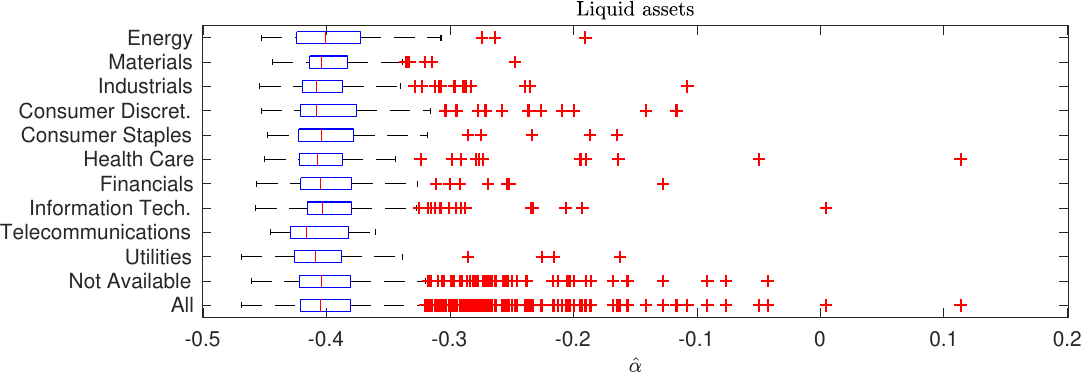}
\caption{\it Box plot for $\hat{\alpha}_{NLLS}^*$  by sector.}
\label{fig:a_box}
\end{figure}

Turning now to persistence properties, Figure \ref{fig:b_box} contains box plots of the estimates of $\beta$, here using the robust parametric estimator $\hat{\beta}_{Cauchy}^*$, based on the Cauchy model. We \mikkolll{have} obtained similar results using the estimator $\hat{\beta}_{BSS}^*$ derived from the Power-$\BSS$ process but for brevity these results are omitted. Again, our analysis confirms the findings seen earlier \mikkolll{with} the E-mini S\&P 500 volatility data. Indeed, the estimates show that log volatility is very persistent, with estimates of $\beta$ \mikkolll{predominantly} in the interval $(0,0.4)$. In particular, there is compelling evidence of very strong persistence in volatility.

\begin{figure}[!t] 
\centering 
\includegraphics[scale=0.9]{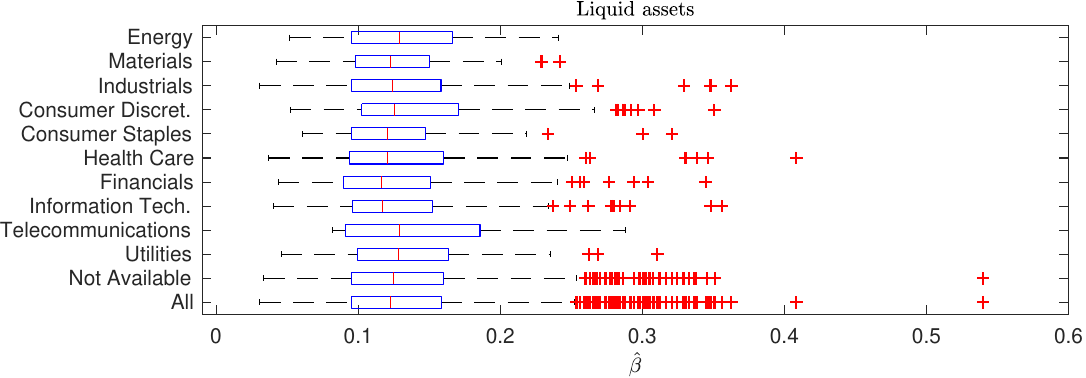}
\caption{\it Box plot for $\hat{\beta}_{Cauchy}^*$ by sector. Bandwidth $H = \lceil n^{1/3} \rceil$.}
\label{fig:b_box}
\end{figure}

The above analysis was conducted using the noise-robust estimators $\hat{\alpha}_{NLLS}^*$ and $\hat{\beta}_{Cauchy}^*$. For completeness, we present the analogous results using the non-noise-robust estimators  $\hat{\alpha}_{OLS}$ and $\hat{\beta}_{Cauchy}$ in the Web Appendix. The findings using these estimators are very similar to those shown in Figures \ref{fig:a_box} and \ref{fig:b_box}.

In summary, we find that roughness and persistence of log volatility extends \mikkolll{to the level of individual equities}. \mikkolll{By and large}, \mikko{they} appear to be \emph{universal} properties of the (logarithm of) realized measures of volatility of equities. (Again \mikko{with} the caveat that these conclusions can, strictly speaking, only be made about the observable quantities, i.e.\mikkolll{,} the realized measures, and not about the latent spot volatility process itself.)

\subsubsection{Roughness and liquidity}\label{sec:Liq}
Figure \ref{fig:aTime} examined the evolution of estimates of the roughness parameter $\alpha$, obtained from the E-mini S\&P 500 volatility data, over time. The degree of roughness appears to be \mikkolll{variable} in time and less rough periods seem to be \mikkolll{concordant with} periods of market turmoil. Given these findings on time-varying roughness, we conjecture that there is a connection between how \mikkolll{actively} an asset is traded and how rough its volatility is. In particular, we expect the volatility of a highly liquid asset to be smoother than the volatility of a less liquid asset. To \mikkolll{briefly explore this idea}, Figure \ref{fig:a_trade} plots the noise-robust NLLS estimates of $\alpha$ against a measure of liquidity,  the logarithm of the average daily volume. The fitted line indicates an increasing pattern (the slope is $0.0091$ with a standard deviation of $0.001$), lending support to the conjecture. \mikkolll{While a compelling explanation for this finding remains to be discovered in future work, it is plausible that market microstructure is at play here. In fact, \cite{JR16} and \cite{EFR16} have recently developed a theoretical framework in which trade execution algorithms that split large trades, that is, \emph{parent orders} (or \emph{metaorders}) into schedules of multiple smaller \emph{child orders} can give rise to rough volatility on aggregate. It is also worth mentioning the recent paper by \cite{GH2019}, where they demonstrate, influenced by an earlier version of the present paper, that a factor investment strategy that uses roughness of volatility as a factor, has had very favorable performance during the recent years. }

\begin{figure}[!t] 
\centering 
\includegraphics[scale=0.95]{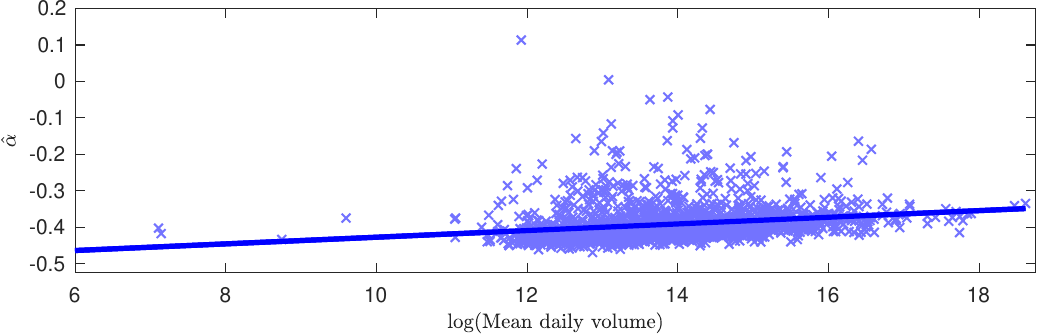}
\caption{\it $\hat{\alpha}^*_{NLLS}$ as a function of average log(mean daily volume). The crosses correspond to individual assets in the data set. The line has been fitted by OLS.}
\label{fig:a_trade}
\end{figure}

\section{Application to volatility forecasting}\label{sec:forecast}
In this section we apply the $\BSS$ and Cauchy models to forecast intraday volatility of the E-mini S\&P $500$ futures contract, comparing the results with a number of benchmark models. The benchmark models can roughly be divided into three categories:
\begin{enumerate}[label=(\roman*),ref=\roman*,leftmargin=3em]
\item\label{cat:standard} standard models,
\item\label{cat:persist} highly persistent models (possibly with long memory),
\item\label{cat:rough} rough volatility models.
\end{enumerate}
The category \eqref{cat:standard} consists of the random walk (RW), autoregressive models (AR) \mik{and an ARMA$(1,1)$ model; \eqref{cat:persist} of the (log) \emph{heterogeneous autoregressive} (log-HAR) model of \cite{corsi09} as well as two ARFIMA models; \eqref{cat:rough} contains only the rough fractional stochastic volatility (RFSV) model of \cite{GJR14}. As for the models suggested in this paper, we consider the Cauchy process and the Power-$\BSS$ and Gamma-$\BSS$ models. As discussed above, these three processes all decouple long- and short-term behavioral characteristics, so with suitable parameter values, they can be seen as members of both \eqref{cat:persist} and \eqref{cat:rough}.}\footnote{
While the $\BSS$ models are initially defined in \emph{continuous time} for the log volatility process $X_t = \log (\sigma_t/\xi)$, $t \geq 0$, cf.\ equation \eqref{eq:sig}, we can apply them, or rather their correlation structure, here in \emph{discrete time} to forecast IV via its estimate $BV^{\Delta*}$, as motivated by the approximation of $\sigma_t$ by $\hat{\sigma}^2_t = \Delta^{-1} \widehat{IV}_t^{\Delta}$, cf.\ equation \eqref{eq:spotvol}.
}

The ARMA and ARFIMA models \mikkolll{have been} estimated by maximum likelihood.\footnote{The ARMA model \mikkolll{has been} estimated and forecasted using the MFE toolbox of Kevin Sheppard, see: \url{https://www.kevinsheppard.com}. The ARFIMA models \mikkolll{has been} estimated and forecasted using the MATLAB package ``ARFIMA(p,d,q) estimator" available from MATLAB Central.} The AR models, as well as the log-HAR model, \mikkolll{have been} estimated by OLS. We apply an intraday adaptation of the standard log-HAR model. The time periods used in constructing the HAR regressors are still one day, one week, and one month, respectively, but the regressors need to be adapted to the step size $\Delta$, which may now be less than a day. More precisely, our log-HAR regression is
\begin{align}
\log \left( BV^{\Delta*}_{t+\Delta h} \right) &= a_0 + a_1 \log \left( BV^{\Delta*}_{t} \right) + a_2 \log \left( BV_{t}^{\Delta*,\textnormal{day}} \right) + a_3 \log \left( BV_{t}^{\Delta*,\textnormal{week}}  \right)  \label{eq:logHAR} \\
&+  a_4 \log \left( BV_{t}^{\Delta*,\textnormal{month}} \right) + \epsilon_{t+\Delta h}, \nonumber
\end{align}
where
\begin{align*}
BV_{t}^{\Delta*,x} := \frac{1}{q} \sum_{k=0}^{q-1} BV^{\Delta*}_{t-k\Delta}, \quad x = \textnormal{day}, \textnormal{week}, \textnormal{month},
\end{align*}
and $q$ is an integer such that $q\Delta = x$. For instance, when $\Delta = 65$ minutes then $q = 6$ for $x = \textnormal{day}$ ($6$ periods of $65$ minutes during a trading day) whereas $q = 30$ for $x = \textnormal{week}$ ($5$ trading days, each consisting of $6$ periods, in a week).\footnote{Although not indicated here, in the variables entering the log-HAR model, all pre-averaged estimates of integrated variance are de-seasonalized when we estimate and forecast the model for $\Delta <$ 1 day. The seasonal factor is re-introduced after estimation and forecasting.}

The estimate of the Hurst index $H$ in the RFSV model is derived via $\hat{H} = \hat{\alpha}_{OLS} + 0.5$, and for the Cauchy and $\BSS$ processes we use the parametric estimates of $\beta$, $\gamma$, and $\lambda$ along with $\hat{\alpha}_{OLS}$. Forecasting the AR and log-HAR models is standard. To forecast the RFSV model, we use the the following approximation \citep[cf.][equation (5.1)]{GJR14}
\begin{multline*}
\E\big[\log \sigma_{t+h\Delta }^2 \big| \mathcal{F}_t\big] \\
\begin{aligned}
& \approx \frac{\cos(H\pi)}{\pi} (h\Delta )^{H+1/2} \int_{-\infty}^t \frac{\log \sigma_s^2}{(t-s+ h\Delta)(t-s)^{H+1/2}}ds \\
	&=  \frac{\cos(H\pi)}{\pi} (h\Delta )^{H+1/2} \sum_{j=1}^{\infty} \int_{t- j\Delta}^{t-(j-1)\Delta} \frac{ \log \sigma_{s}^2}{(t-s+ h\Delta)(t-s)^{H+1/2}}ds \\
	&\approx  \frac{\cos(H\pi)}{\pi} (h\Delta )^{H+1/2} \sum_{j=1}^{\infty} \log \sigma_{t-(j-1)\Delta}^2 \int_{t- j\Delta}^{t-(j-1)\Delta} \frac{1}{(t-s+ h\Delta)(t-s)^{H+1/2}}ds,
\end{aligned}
\end{multline*}
where $\mathcal{F}_t$ is the information set ($\sigma$-algebra) generated by the fBm driving the model up to time $t$. The integrals are approximated by Riemann sums.

To forecast the Cauchy and $\BSS$ processes we rely on the elementary result that for a zero-mean Gaussian random vector $(x_{t+h}, x_t, x_{t-1}, \ldots, x_{t-m})^T$ the distribution of $x_{t+h}$ conditionally on $(x_t, x_{t-1}, \ldots, x_{t-m})^T = a \in \R^{m+1}$ is
\begin{align*}
x_{t+h} | \left\{(x_t, x_{t-1}, \ldots, x_{t-m})^T = a \right\}\sim N(\mu,\xi^2),
\end{align*}
where
\begin{align*}
\mu = \Gamma_{12} \Gamma_{22}^{-1} a,
\end{align*}
with $\Gamma_{22}$ being the correlation matrix of the vector $(x_t, x_{t-1}, \ldots, x_{t-m})^T$, and
\begin{align*}
\Gamma_{12} := \left(Corr(x_{t+h},x_t),Corr(x_{t+h},x_{t-1}) , \ldots, Corr(x_{t+h},x_{t-m})\right).
\end{align*}

Since the processes we consider here are stationary, the variance $\xi^2$ of the conditional distribution is
\begin{align*}
\xi^2 = Var(x_t) \left(1-\Gamma_{12} \Gamma_{22}^{-1}\Gamma_{21}\right),
\end{align*}
where $\Gamma_{21} = \Gamma_{12}^T$.

To implement this procedure for the Cauchy and $\BSS$ models, we assume Gaussianity of the process and use these results, where the correlation matrices and vectors above are calculated from the theoretical correlation structure of the process in question, implied by the estimated parameters. When forecasting log volatility, only the conditional mean $\mu$ needs to be calculated. However, as we will argue in the next section, the conditional variance term $\xi^2$ will be important in forecasting (non-logarithmic) volatility. These results rely on $X$ having mean zero, so in our forecasting experiment we de-mean the data before conducting the experiment.\footnote{In our model for volatility, this de-meaning essentially means removing the term \mikkolll{$\log \xi = \mu$}, cf.\ equation \eqref{eq:sig}. We reintroduce this term after forecasting $X$.}

\subsection{Forecasting intraday integrated variance}\label{sec:predRaw}
The previous section lays out methods of forecasting log volatility or, equivalently, the process $X$, cf. equation \eqref{eq:sig}. However, it is in practice more relevant to forecast non-logarithmic volatility. Before presenting the forecasting results for this quantity, we briefly explain our approach.

As we are now interested in $\E [\exp( X_{t+\Delta}) |\F_{t}]$, instead of $\exp( \E [ X_{t+\Delta} |\F_{t}])$, it is worth reminding that it is a flawed strategy to simply forecast log volatility as above and then exponentiate the forecast. Indeed, by Jensen's inequality we know this approach to be biased.  However, we can often correct the exponentiated forecasts following a simple approach. For the $\BSS$ and Cauchy models we follow the strategy of the preceding section. That is, if we again assume Gaussianity, we have
\begin{align}\label{eq:condExpExp}
\E [\exp( X_{t+\Delta}) |\F_{t}] = \exp \left( \E [X_{t+\Delta} |\F_{t}] + \frac{1}{2}Var[X_{t+\Delta} |\F_{t}]\right).
\end{align}
We will then approximate the former term in the exponential function by $\mu$ and the latter by $\frac{1}{2}\xi^2$. Note that $\xi^2$ depends on the (stationary) variance of the process, $Var(x_t)$; this factor we simply estimate from the (unconditional) variance of the time series being forecasted.

As for the other models, \cite{GJR14} proposed a similar correction to their RFSV model \citep[see][Section 5.2]{GJR14}, which we use in the following. As the log-HAR model is estimated by OLS using the log $BV^{\Delta*}$ data, cf.\ equation \eqref{eq:logHAR}, we exponentiate these estimates and make a correction similar to \eqref{eq:condExpExp}, where the variance factor is estimated as the variance of the error term in the OLS regression \eqref{eq:logHAR}. The remaining models are estimated   directly using the raw (de-seasonalized) $BV^{\Delta*}$ data, so no correction is needed.

\subsubsection{Forecast setup}

We use the methodology described above to forecast integrated variance, as this is most often the object of interest in applications. Since integrated variance is not actually observable, as a feasible \emph{forecast object} (FO) we use
\begin{align}\label{eq:FO}
\text{FO}_t(\Delta,h) :=  \sum_{k=1}^{h} BV^{\Delta*}_{t+k\Delta}  \approx \int_t^{t+h\Delta} \sigma_s^2ds, \quad h = 1, 2, 5, 10, 20,
\end{align}
where $BV^{\Delta*}_t$ is the estimated value of integrated variance using the pre-averaged bipower variation estimator, cf.\ Section \ref{sec:empirical}. 

To forecast the FO in \eqref{eq:FO}, we compute the $h$ individual components, $\widehat{\sigma^2}_{t+k\Delta|t}$, $k = 1, \ldots, h$, \mik{multiply} by $\Delta$ and the seasonal component, and sum them up:
\begin{align*}
\widehat{\text{FO}}_t(\Delta,h) = \sum_{k=1}^{h} \widehat{\sigma^2}_{t+k\Delta|t} \left(\sigma_{t+k\Delta}^s\right)^2 \Delta,
\end{align*}
where $\sigma_t^s$ is the (deterministic) seasonal component of volatility we extracted in the preliminary step of our analysis, as explained in Section \ref{sec:empirical}, and $\widehat{\sigma^2}_{t+k\Delta|t}$ is the forecast of volatility, as detailed above.\footnote{To keep the procedure realistic, in the out-of-sample experiment discussed below, the seasonal component is estimated in a non-anticipative fashion, using only the observations that would be available at the time when the forecast is produced.}

In the forecasting experiments we consider various step sizes $\Delta$, ranging from $15$ minutes to $1$ day, and various forecast horizons $h \in \{1,2,5,10,20\}$. We start the estimation after an initial period of $m \in \N$ time steps and compare the performance of the forecasts using two different loss functions:
\begin{itemize}
\item
	Mean Squared Error (MSE): 
\begin{itemize}
\item[]
MSE$(\Delta,h) = \frac{1}{n-h-m+1}\sum_{t=m}^{n-h} \left| \widehat{\text{FO}}_t(\Delta,h)  - \text{FO}_t(\Delta,h) \right|^2$,
\end{itemize}
\item
	\mikkelll{``Quasi-likelihood" (QLIKE):}
\begin{itemize}
\item[]
QLIKE$(\Delta,h)= \frac{1}{n-h-m+1}\sum_{t=m}^{n-h} \left(\log \widehat{\text{FO}}_t(\Delta,h) + \frac{\text{FO}_t(\Delta,h)}{\widehat{\text{FO}}_t(\Delta,h)}\right)$.
\end{itemize}
\end{itemize}

As discussed in Section \ref{sec:empirical}, the the pre-averaged estimate $BV^{\Delta*}$, our FO, is a noisy estimate of integrated variance, but \cite{patton11} shows that the MSE and QLIKE loss functions still yield consistent rankings of the forecasting models even for integrated variance, in spite of the noisy estimates used to evaluate the loss functions. We calculate MSE, QLIKE, and also the \emph{model confidence set} (MCS)  of \cite{MCS}, which is a procedure to construct a ``set of best models" with a certain probability, as measured by the specific loss function in question, \mikkolll{avoiding} the problems that arise from doing multiple comparisons by pairwise tests. For instance, the \emph{best} model is contained in the $90\%$ MCS --- when understood as a random set ---  with  $90\%$ probability. 

We set up our forecasting experiment so that it is realistic and mirrors the situation a practitioner would face if they were to forecast intraday volatility. At time $t$, we use only the data observed so far to estimate the models and forecast $h$ steps ahead, for $h = 1, 2, 5, 10, 20$. We then move one step forward in time, to $t+\Delta$, re-estimate the models, to allow for time variation in the parameters, and compute new forecasts $h$ steps ahead. In re-estimation, we used a rolling window of 200 observations. While this window length is somewhat arbitrary, brief experimentation with other window lengths \mikkolll{suggests} that our results are not particularly dependent on this choice. We note that, since this forecasting exercise entails that the parameters of the various models need to be re-estimated many times, we found it infeasible to use the NLLS estimator of $\alpha$ here, and we therefore use the OLS estimator instead. As shown in Table \ref{tab:estBoth}, these two estimators often yield similar estimates when applied to this data set, so we expect the OLS estimator to be adequate for the purposes of this forecasting exercise.

We assess the forecasting performance of the different models over two time periods using the E-mini S\&P $500$ data set. First we forecast over a long out-of-sample period from from January 3, 2005 to December 31, 2014, using the the entire data set at our disposal (cf.\ Section \ref{sec:time}, where we presented evidence of the roughness parameter varying in time). This period contains a couple of abrupt market crashes \mikkolll{(e.g., the bankruptcy of Lehman Brothers on September 15, 2008 and the Flash Crash on May 6, 2010)}, precipitated by circumstances that would have been hard to predict by virtually any model based on historical \mikkolll{volatility} data \mikkolll{only}. We observed that failing to forecast the volatility bursts during such episodes leads to disproportionately large losses with both MSE and QLIKE, and especially with MSE, these individual losses can represent a significant proportion of the total loss. This phenomenon also affects the MCS procedure, which would not be very selective over this long out-of-sample period. For this reason, we also experiment with another, shorter out-of-sample period from from January 2, 2013 to December 31, 2014. Volatility over this period is less pronounced, exemplifying ``non-stressed'' market conditions.

\subsubsection{Results of out-of-sample forecasting experiment}

Figure \ref{fig:fig1} plots the cumulative QLIKE losses for select models in the case of forecasting intraday ($\Delta = 15$ minutes) integrated variance over the long out-of-sample period from 2005 to 2014. We employ the QLIKE loss function since it is less prone to exaggerating the impact of large forecast errors than MSE \citep{patton11}. We present our results on the cumulative losses relative to those arising when the Gamma-$\BSS$ model is used. More precisely, the graphs present, as a function of time $t$,
\begin{align*}
\textnormal{Cumulative QLIKE (vs. Gamma-BSS)}^x_t := \sum_{k=1}^{\lfloor t/\Delta \rfloor}\left(\textnormal{QLIKE}(x)_k - \textnormal{QLIKE(Gamma-$\BSS$)}_k\right),
\end{align*}
where $x \in \{ \textnormal{HAR}, \textnormal{RFSV}, \textnormal{Cauchy}, \textnormal{Power-BSS}\}$ denotes the model being compared to the Gamma-$\BSS$ model, and $\textnormal{QLIKE}(x)_k$ stands for the $k$'th loss using model $x$. With this definition, positive numbers indicate that the model performs worse than the Gamma-$\BSS$ model and vice versa for negative numbers.

We observe that both the Gamma-$\BSS$ and especially the Power-$\BSS$ models perform well over the entire period. Interestingly, as the forecasting horizon $h$ is increased, the Cauchy model performs better relative to the Gamma-$\BSS$ model, suggesting that the long memory property of the Cauchy model, which is absent from the Gamma-$\BSS$ model, becomes relevant.
\begin{figure}[!t] 
\centering 
\includegraphics[scale=0.9]{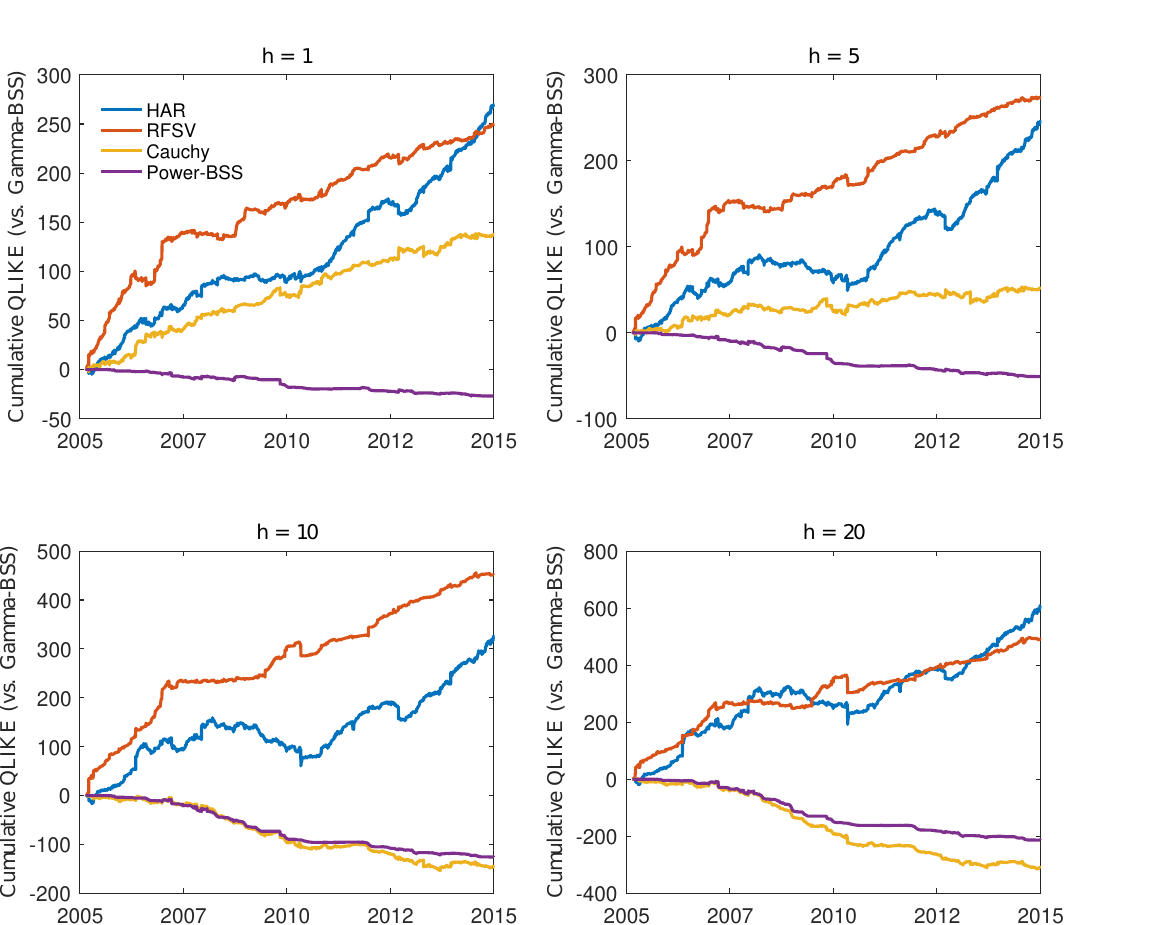}
\caption{\it Cumulative QLIKE forecast errors through time. Here, $h$ is the forecast horizon (i.e., volatility is forecasted over a time interval of length $h\Delta$) and $\Delta = 15$ minutes.}
\label{fig:fig1}
\end{figure}

Turning now to the shorter out-of-sample period from January 2, 2013 to December 31, 2014 that represents calm market conditions, we present more detailed forecasting results for $\Delta = 15, 30, 65$ minutes and $\Delta = 1$ day in Tables \ref{tab:predRawOutES1} and \ref{tab:predRawOutES2}. There, a boldface number identifies the model with the smallest loss, while light grey background highlights the models in the $90\%$ MCS and dark grey background those in the $75\%$ MCS. The models proposed in this paper generally outperform the benchmarks over this period. Indeed, in most cases either the Cauchy or the Power-$\BSS$ model has the lowest MSE and QLIKE losses.  Further, these two models are very often in the $75\%$ MCS. These findings are most convincing for the intraday values of $\Delta$, while, interestingly, RFSV turns out to be the best performing model when $\Delta = 1$ day. (Unfortunately, the MCS is not very selective in this case, at least when using MSE as the loss function.) This result agrees with \cite{GJR14}, where the authors demonstrated that the RFSV model can outperform a range of benchmark models when forecasting daily integrated variance one step ahead. 

We remark on the perhaps surprising observation that including a long memory component seems to also improve short term forecasts. Further research into the reasons for this would be interesting and valuable. All in all, the results presented in this section suggest that it is in general advantageous to exploit the roughness of volatility in forecasting, and at intraday time scales, careful modeling of persistence can further improve the quality of forecasts.

\begin{landscape}

\begin{table*}
\caption{\it Out-of-sample forecasting of intraday \mikko{integrated variance}}
\begin{center}
\scriptsize 
 \begin{tabularx}{1.35\textwidth}{@{\extracolsep{\stretch{1}}}lc@{\hskip -0.2in}cc@{\hskip -0.2in}cc@{\hskip -0.2in}cc@{\hskip -0.2in}cc@{\hskip -0.2in}c@{}} 
 \multicolumn{11}{l}{Panel A: $\Delta = 15$ minutes}\\
\toprule
 & \multicolumn{2}{c}{$h = 1$} & \multicolumn{2}{c}{$h = 2$} & \multicolumn{2}{c}{$h = 5$} & \multicolumn{2}{c}{$h = 10$} & \multicolumn{2}{c}{$h = 20$}  \\ 
\cmidrule{2-11}
 & MSE & QLIKE & MSE & QLIKE  & MSE & QLIKE  & MSE & QLIKE  & MSE & QLIKE  \\ 
  & $\times 10^{11}$  &   & $\times 10^{11}$  &   & $\times 10^{10}$  &   & $\times 10^{9}$  &   & $\times 10^{8}$  &   \\ 
 \midrule 
RW & $  0.229\MCStwo$ & $ -12.883 $ & $  0.733 $ & $ -12.153 $ & $  0.471 $ & $ -11.147 $ & $  0.245 $ & $ -10.355 $ & $  0.102 $ & $ -9.553 $\\
AR5 & $  0.178\MCStwo$ & $ -12.928 $ & $  0.472\MCStwo$ & $ -12.214 $ & $  0.220 $ & $ -11.252 $ & $  0.079 $ & $ -10.509 $ & $  0.027 $ & $ -9.754 $\\
AR10 & $  0.187\MCStwo$ & $ -12.931 $ & $  0.494 $ & $ -12.217 $ & $  0.226 $ & $ -11.257 $ & $  0.077 $ & $ -10.514 $ & $  0.027 $ & $ -9.760 $\\
ARMA(1,1) & $  0.159\MCStwo$ & $ -12.932 $ & $  0.433 $ & $ -12.217 $ & $  0.203 $ & $ -11.255 $ & $  0.074 $ & $ -10.509 $ & $  0.027 $ & $ -9.748 $\\
log-HAR3 & $  0.149 $ & $ -12.941 $ & $  0.408 $ & $ -12.225 $ & $  0.198 $ & $ -11.260 $ & $  0.079 $ & $ -10.512 $ & $  0.038 $ & $ -9.745 $\\
ARFIMA$(0,d,0)$ & $  0.145 $ & $ -12.937 $ & $  0.393 $ & $ -12.222 $ & $  0.176 $ & $ -11.262 $ & $  0.059 $ & $ -10.522 $ & $  0.021 $ & $ -9.772 $\\
RFSV & $  0.145\MCSone$ & $ -12.947 $ & $  0.384\MCSone$ & $ -12.232 $ & $  0.170\MCStwo$ & $ -11.268 $ & $  0.057 $ & $ -10.525 $ & $  0.020 $ & $ -9.774 $\\
Cauchy & $  0.143\MCStwo$ & $ -12.948 $ & $  0.382\MCStwo$ & $ -12.234\MCStwo$ & $  0.164\MCSone$ & $ -11.274\MCSone$ & $  0.053\MCSone$ & $ -10.533\MCSone$ & $\mathbf{ 0.019}\MCSone $ & $\mathbf{-9.783}\MCSone $\\
Power-BSS & $\mathbf{ 0.142}\MCSone $ & $\mathbf{-12.950}\MCSone $ & $\mathbf{ 0.377}\MCSone $ & $\mathbf{-12.235}\MCSone $ & $\mathbf{ 0.163}\MCSone $ & $\mathbf{-11.274}\MCSone $ & $\mathbf{ 0.053}\MCSone $ & $\mathbf{-10.533}\MCSone $ & $  0.019\MCSone$ & $ -9.783\MCSone$\\
Gamma-BSS & $  0.142\MCSone$ & $ -12.949 $ & $  0.377\MCStwo$ & $ -12.234 $ & $  0.164 $ & $ -11.273 $ & $  0.054 $ & $ -10.531 $ & $  0.019 $ & $ -9.780 $\\
\bottomrule 
\end{tabularx}
 \begin{tabularx}{1.35\textwidth}{@{\extracolsep{\stretch{1}}}lc@{\hskip -0.2in}cc@{\hskip -0.2in}cc@{\hskip -0.2in}cc@{\hskip -0.2in}cc@{\hskip -0.2in}c@{}} 
 \multicolumn{11}{l}{Panel B: $\Delta = 30$ minutes}\\
\toprule
 & \multicolumn{2}{c}{$h = 1$} & \multicolumn{2}{c}{$h = 2$} & \multicolumn{2}{c}{$h = 5$} & \multicolumn{2}{c}{$h = 10$} & \multicolumn{2}{c}{$h = 20$}  \\ 
\cmidrule{2-11}
 & MSE & QLIKE & MSE & QLIKE  & MSE & QLIKE  & MSE & QLIKE  & MSE & QLIKE  \\ 
  & $\times 10^{11}$  &   & $\times 10^{10}$  &   & $\times 10^{9}$  &   & $\times 10^{9}$  &   & $\times 10^{8}$  &   \\ 
 \midrule 
RW & $  0.596 $ & $ -12.181 $ & $  0.220 $ & $ -11.435 $ & $  0.170 $ & $ -10.404 $ & $  0.612 $ & $ -9.617 $ & $  0.264 $ & $ -8.874 $\\
AR5 & $  0.500 $ & $ -12.213 $ & $  0.154 $ & $ -11.490 $ & $  0.074 $ & $ -10.512 $ & $  0.259 $ & $ -9.758 $ & $  0.113 $ & $ -9.037 $\\
AR10 & $  0.501 $ & $ -12.214 $ & $  0.151 $ & $ -11.491 $ & $  0.073 $ & $ -10.513 $ & $  0.262 $ & $ -9.761 $ & $  0.117 $ & $ -9.043 $\\
ARMA(1,1) & $  0.459 $ & $ -12.215 $ & $  0.142 $ & $ -11.492 $ & $  0.073 $ & $ -10.511 $ & $  0.261 $ & $ -9.753 $ & $  0.119 $ & $ -9.023 $\\
log-HAR3 & $  0.433\MCSone$ & $ -12.227 $ & $  0.131\MCStwo$ & $ -11.501 $ & $  0.070\MCStwo$ & $ -10.517 $ & $  0.300\MCStwo$ & $ -9.762 $ & $  0.160 $ & $ -9.032 $\\
ARFIMA$(0,d,0)$ & $  0.432 $ & $ -12.222 $ & $  0.129 $ & $ -11.498 $ & $  0.062 $ & $ -10.522 $ & $  0.218 $ & $ -9.772 $ & $  0.088 $ & $ -9.055 $\\
RFSV & $  0.438\MCSone$ & $ -12.230\MCSone$ & $  0.128\MCSone$ & $ -11.504\MCStwo$ & $  0.060\MCStwo$ & $ -10.523 $ & $  0.206\MCSone$ & $ -9.772 $ & $  0.089\MCSone$ & $ -9.059 $\\
Cauchy & $  0.418\MCSone$ & $ -12.229\MCSone$ & $  0.122\MCSone$ & $ -11.506\MCSone$ & $\mathbf{ 0.057}\MCSone $ & $\mathbf{-10.530}\MCSone $ & $\mathbf{ 0.199}\MCSone $ & $\mathbf{-9.781}\MCSone $ & $\mathbf{ 0.080}\MCSone $ & $\mathbf{-9.067}\MCSone $\\
Power-BSS & $\mathbf{ 0.416}\MCSone $ & $\mathbf{-12.231}\MCSone $ & $\mathbf{ 0.122}\MCSone $ & $\mathbf{-11.506}\MCSone $ & $  0.057\MCSone$ & $ -10.529\MCStwo$ & $  0.199\MCSone$ & $ -9.780 $ & $  0.081\MCSone$ & $ -9.065 $\\
Gamma-BSS & $  0.418\MCStwo$ & $ -12.230 $ & $  0.124\MCStwo$ & $ -11.504 $ & $  0.059\MCStwo$ & $ -10.526 $ & $  0.205\MCStwo$ & $ -9.776 $ & $  0.085\MCStwo$ & $ -9.059 $\\
\bottomrule 
\end{tabularx}
\end{center}
{\footnotesize \it Out-of-sample Mean Squared Forecast Error (MSE) and QLIKE for all models considered in the paper. Forecasting period: January 2, 2013, to December 31, 2014. Bold numbers \mikko{indicate} the model with the smallest forecast error (column-wise). The forecast \mikko{object} is the sum of realized kernel\mikko{s \eqref{eq:FO}}, approximating \mikko{integrated variance}, as explained in the text. \mikko{We vary the step size $\Delta$ and the forecast horizon $h$.} Grey cells \mikko{indicate} models which are in the Model Confidence Set (column-wise); the dark grey denotes the $75\%$ MCS, while the light grey denotes the $90\%$ MCS. The MCS uses a block bootstrap method with $25\ 000$ bootstrap replications and a block length of $6$ \mikko{time steps}.}
\label{tab:predRawOutES1}
\end{table*}
  
\end{landscape}

\begin{landscape}

\begin{table*}
\caption{\it Out-of-sample forecasting of intraday \mikko{integrated variance}}
\begin{center}
\scriptsize 
 \begin{tabularx}{1.35\textwidth}{@{\extracolsep{\stretch{1}}}lc@{\hskip -0.2in}cc@{\hskip -0.2in}cc@{\hskip -0.2in}cc@{\hskip -0.2in}cc@{\hskip -0.2in}c@{}} 
 \multicolumn{11}{l}{Panel A: $\Delta = 65$ minutes}\\
\toprule
 & \multicolumn{2}{c}{$h = 1$} & \multicolumn{2}{c}{$h = 2$} & \multicolumn{2}{c}{$h = 5$} & \multicolumn{2}{c}{$h = 10$} & \multicolumn{2}{c}{$h = 20$}  \\ 
\cmidrule{2-11}
 & MSE & QLIKE & MSE & QLIKE  & MSE & QLIKE  & MSE & QLIKE  & MSE & QLIKE  \\ 
  & $\times 10^{10}$  &   & $\times 10^{9}$  &   & $\times 10^{9}$  &   & $\times 10^{8}$  &   & $\times 10^{7}$  &   \\ 
 \midrule 
RW & $  0.314\MCStwo$ & $ -11.355 $ & $  0.147\MCStwo$ & $ -10.579 $ & $  0.675 $ & $ -9.557 $ & $  0.297 $ & $ -8.822 $ & $  0.127 $ & $ -8.079 $\\
AR5 & $  0.237\MCStwo$ & $ -11.389 $ & $  0.074\MCStwo$ & $ -10.648 $ & $  0.367\MCStwo$ & $ -9.659 $ & $  0.158\MCSone$ & $ -8.947 $ & $  0.078\MCSone$ & $ -8.226 $\\
AR10 & $  0.247\MCStwo$ & $ -11.337\MCStwo$ & $  0.076\MCStwo$ & $ -10.615\MCStwo$ & $  0.378\MCStwo$ & $ -9.661 $ & $  0.178\MCSone$ & $ -8.949 $ & $  0.093\MCSone$ & $ -8.226 $\\
ARMA(1,1) & $  0.221\MCStwo$ & $ -11.390 $ & $  0.071\MCStwo$ & $ -10.649 $ & $  0.347 $ & $ -9.656 $ & $  0.154\MCStwo$ & $ -8.935 $ & $  0.078\MCSone$ & $ -8.206 $\\
log-HAR3 & $  0.207\MCSone$ & $ -11.404\MCStwo$ & $  0.067\MCStwo$ & $ -10.657 $ & $  0.370\MCSone$ & $ -9.662 $ & $  0.178\MCSone$ & $ -8.938 $ & $  0.104\MCStwo$ & $ -8.189 $\\
ARFIMA$(0,d,0)$ & $  0.205 $ & $ -11.397 $ & $  0.063 $ & $ -10.658 $ & $  0.302 $ & $ -9.673 $ & $  0.116\MCStwo$ & $ -8.960 $ & $  0.047\MCSone$ & $ -8.240 $\\
RFSV & $  0.201\MCSone$ & $\mathbf{-11.409}\MCSone $ & $  0.062\MCStwo$ & $ -10.664 $ & $  0.281\MCSone$ & $ -9.676 $ & $  0.109\MCSone$ & $ -8.966\MCSone$ & $  0.047\MCSone$ & $\mathbf{-8.250}\MCSone $\\
Cauchy & $  0.199\MCSone$ & $ -11.406\MCStwo$ & $  0.059\MCSone$ & $ -10.668\MCSone$ & $  0.285\MCSone$ & $\mathbf{-9.683}\MCSone $ & $  0.108\MCSone$ & $\mathbf{-8.970}\MCSone $ & $\mathbf{ 0.043}\MCSone $ & $ -8.248\MCSone$\\
Power-BSS & $\mathbf{ 0.195}\MCSone $ & $ -11.409\MCSone$ & $\mathbf{ 0.058}\MCSone $ & $\mathbf{-10.668}\MCSone $ & $\mathbf{ 0.279}\MCSone $ & $ -9.682\MCSone$ & $\mathbf{ 0.107}\MCSone $ & $ -8.970\MCSone$ & $  0.044\MCSone$ & $ -8.247\MCSone$\\
Gamma-BSS & $  0.197\MCSone$ & $ -11.407\MCStwo$ & $  0.059\MCSone$ & $ -10.665 $ & $  0.282\MCSone$ & $ -9.679 $ & $  0.109\MCSone$ & $ -8.967 $ & $  0.046\MCSone$ & $ -8.242 $\\
\bottomrule 
\end{tabularx}
 \begin{tabularx}{1.35\textwidth}{@{\extracolsep{\stretch{1}}}lc@{\hskip -0.2in}cc@{\hskip -0.2in}cc@{\hskip -0.2in}cc@{\hskip -0.2in}cc@{\hskip -0.2in}c@{}} 
 \multicolumn{11}{l}{Panel B: $\Delta = 1$ day}\\
\toprule
 & \multicolumn{2}{c}{$h = 1$} & \multicolumn{2}{c}{$h = 2$} & \multicolumn{2}{c}{$h = 5$} & \multicolumn{2}{c}{$h = 10$} & \multicolumn{2}{c}{$h = 20$}  \\ 
\cmidrule{2-11}
 & MSE & QLIKE & MSE & QLIKE  & MSE & QLIKE  & MSE & QLIKE  & MSE & QLIKE  \\ 
  & $\times 10^{9}$  &   & $\times 10^{8}$  &   & $\times 10^{7}$  &   & $\times 10^{6}$  &   & $\times 10^{5}$  &   \\ 
 \midrule 
RW & $  0.614\MCSone$ & $ -9.491 $ & $  0.253\MCStwo$ & $ -8.783 $ & $  0.170\MCSone$ & $ -7.818 $ & $  0.076\MCStwo$ & $ -7.045 $ & $  0.035\MCStwo$ & $ -6.218 $\\
AR5 & $  0.717\MCSone$ & $ -9.531 $ & $  0.326\MCSone$ & $ -8.826 $ & $  0.179\MCSone$ & $ -7.882 $ & $  0.053\MCSone$ & $ -7.146 $ & $  0.013\MCStwo$ & $ -6.405\MCStwo$\\
AR10 & $  0.729\MCSone$ & $ -8.717\MCSone$ & $  0.338\MCStwo$ & $ -7.976\MCSone$ & $  0.185\MCStwo$ & $ -7.876 $ & $  0.055\MCStwo$ & $ -7.142 $ & $  0.014\MCStwo$ & $ -6.403\MCStwo$\\
ARMA(1,1) & $  0.621\MCSone$ & $ -9.537 $ & $  0.266\MCSone$ & $ -8.831 $ & $  0.232\MCSone$ & $ -7.883 $ & $  0.189\MCSone$ & $ -7.146\MCStwo$ & $  0.496\MCStwo$ & $ -6.400\MCSone$\\
log-HAR3 & $  0.506\MCSone$ & $ -9.543\MCSone$ & $  0.179\MCSone$ & $ -8.838\MCSone$ & $  0.096\MCSone$ & $ -7.890\MCSone$ & $  0.034\MCSone$ & $ -7.154\MCSone$ & $  0.011\MCStwo$ & $ -6.403\MCStwo$\\
ARFIMA$(0,d,0)$ & $  0.542\MCSone$ & $ -9.537 $ & $  0.191\MCStwo$ & $ -8.830 $ & $  0.104\MCStwo$ & $ -7.881 $ & $  0.038\MCStwo$ & $ -7.144 $ & $  0.013\MCStwo$ & $ -6.395 $\\
RFSV & $\mathbf{ 0.493}\MCSone $ & $ -9.546\MCSone$ & $\mathbf{ 0.171}\MCSone $ & $\mathbf{-8.843}\MCSone $ & $  0.091\MCSone$ & $\mathbf{-7.898}\MCSone $ & $  0.034\MCSone$ & $\mathbf{-7.160}\MCSone $ & $  0.014\MCStwo$ & $ -6.402\MCSone$\\
Cauchy & $  0.579\MCSone$ & $ -9.537 $ & $  0.195\MCSone$ & $ -8.831 $ & $  0.098\MCSone$ & $ -7.884 $ & $  0.034\MCStwo$ & $ -7.151\MCStwo$ & $  0.011 $ & $ -6.401 $\\
Power-BSS & $  0.545\MCSone$ & $ -9.543\MCSone$ & $  0.186\MCSone$ & $ -8.837\MCSone$ & $  0.095\MCSone$ & $ -7.889\MCStwo$ & $  0.033\MCSone$ & $ -7.154\MCSone$ & $  0.011\MCStwo$ & $ -6.405\MCStwo$\\
Gamma-BSS & $  0.513\MCSone$ & $\mathbf{-9.547}\MCSone $ & $  0.176\MCSone$ & $ -8.842\MCSone$ & $\mathbf{ 0.091}\MCSone $ & $ -7.895\MCSone$ & $\mathbf{ 0.032}\MCSone $ & $ -7.159\MCSone$ & $\mathbf{ 0.010}\MCSone $ & $\mathbf{-6.412}\MCSone $\\
\bottomrule 
\end{tabularx}
\end{center}
{\footnotesize \it Out-of-sample Mean Squared Forecast Error (MSE) and QLIKE for all models considered in the paper. Forecasting period: January 2, 2013, to December 31, 2014. Bold numbers \mikko{indicate} the model with the smallest forecast error (column-wise). The forecast \mikko{object} is the sum of realized kernel\mikko{s \eqref{eq:FO}}, approximating \mikko{integrated variance}, as explained in the text. \mikko{We vary the step size $\Delta$ and the forecast horizon $h$.}  Grey cells \mikko{indicate} models which are in the Model Confidence Set (column-wise); the dark grey denotes the $75\%$ MCS, while the light grey denotes the $90\%$ MCS. The MCS uses a block bootstrap method with $25\ 000$ bootstrap replications and a block length of $6$ \mikko{time steps}.} 
\label{tab:predRawOutES2}
\end{table*}
  
\end{landscape}

\section{Conclusions}\label{sec:concl}
This paper has presented stochastic models that are able to simultaneously capture roughness and persistence, which are two key empirical features found in realized volatility data. In particular, when modeling log volatility, we advocate using a stochastic process that decouples short- and long-term behavior. The Brownian semistationary process is a flexible and parsimonious example of such a process; since it also allows for easy inclusion of non-Gaussianity and the leverage effect, we argue that the Brownian semistationary process is a particularly good candidate model for stochastic log volatility.

Using our proposed two-step estimation approach, we \mikkolll{have} also presented a thorough investigation of the empirical characteristics of the volatility of the E-mini S\&P \mikkolll{500} futures contract, both \mikkolll{at daily and shorter time scales}. We find that realized measures of volatility of this contract are \mikkolll{indeed rough but also  highly persistent}. Moreover, by also looking at realized volatility data on almost two thousand individual US equities, we corroborated these findings, suggesting that both roughness and strong persistence are universal features of these types of data.

\mikkolll{It is important to note that, since realized measures are noisy estimates of integrated variance, from which working out spot volatility is a numerically ill-posed problem, our findings need not imply that the elusive spot volatility process is actually rough or persistent. However, we have demonstrated extensively that models that incorporate these two characteristics are overwhelmingly \emph{consistent} with the observable statistical features of high-frequency volatility. We have also shown that these models are capable of producing competitive volatility forecasts at daily and shorter time scales. Therefore, irrespective of whether one accepts that volatility is actually rough, we believe that ideas and models stemming from rough volatility offer a valuable alternative viewpoint to volatility modelling, and a useful set of actionable tools, complementing the existing models based on jumps and conventional stochastic volatility.}

\section*{Acknowledgements}
Previous versions of this work have been presented at the \emph{9th Annual SoFiE Conference} in Hong Kong in June 2016, at the ICMS workshop \emph{At the Frontiers of Quantitative Finance} in Edinburgh in June 2016, at the \emph{9th Bachelier World Congress} in New York in July 2016, at the \emph{New Developments in Measuring \& Forecasting Financial Volatility} conference at Duke University, Durham, NC in September 2016, at the \emph{5th Imperial--ETH Workshop on Mathematical Finance} in London in March 2017, and at the Barcelona GSE Summer Forum workshop \emph{Fractional Brownian Motion and Rough Models} in Barcelona in June 2017; and we thank the audiences for stimulating comments and questions. We also thank Jim Gatheral, Sherry Jiang, Mark Podolskij, Roberto Ren\`o, Peter Tankov, Erik Vogt, as well as two anonymous referees, and the Editor, Fabio Trojani, for helpful remarks. Our research has been partially supported by CREATES (DNRF78), funded by the Danish National Research Foundation,  by Aarhus University Research Foundation (project ``Stochastic and Econometric Analysis of Commodity Markets"), and by the Academy of Finland (project 258042).

{\small 
\bibliographystyle{chicago}
\bibliography{mybib-v7}
}

\appendix

\section{Pre-averaged measures of integrated variance}\label{app:BV}
This section follows the methodology in \cite{COP14} closely. Suppose we want to estimate IV in some interval $[(i-1)\Delta,i\Delta]$ for $i\geq 1$ and we have $N+1$ observations, $Z_0, Z_1, \ldots, Z_N$, where $Z_i = \log S_{t_i}$,  of the log price process $Z = \log S$ in this interval. We define the \emph{pre-averaged log returns},
\begin{align*}
\mikkell{r_{j,K}^* = \frac{1}{K} \left( \sum_{k=K/2}^{K-1}Z_{(j+k)} - \sum_{k=0}^{K/2-1} Z_{(j+k)}  \right), \quad j = 0, 1, \ldots, N-K,}
\end{align*}
\mikkell{where $K\geq 2$ is even. For the asymptotics to work, it is required that $K = \theta \sqrt{N} + o\left(N^{-1/4}\right)$, and in our implementation we set $\theta = 1$ and $K = \lfloor \sqrt{N} \rfloor$ if  $\lfloor \sqrt{N} \rfloor$ is an even number and  $K = \lfloor \sqrt{N} \rfloor + 1$ otherwise.\footnote{\mikkell{Setting the tuning parameter $\theta$ equal to $1$ was found to work well for data similar to ours in \cite{COP14}. We \mikkol{drew} the same conclusion \mikkol{from both simulated data mimicking our setup, as well from the actual data analyzed later in the paper.}}} Here, $\lfloor x \rfloor$ means the largest integer smaller than, or equal to, $x\in \R.$ Using these pre-averaged returns, we suggest the following estimators of IV, which are robust to market microstructure noise:}
\begin{align}
\mikkell{RV^{\Delta*}_t }&\mikkell{=  \frac{N}{N-K+2} \frac{1}{K \psi_K} \sum_{j=0}^{N-K+1} |r_{j,K}^*|^2 - \frac{\hat{\omega}^2}{\theta^2 \psi_K}, } \nonumber \\
\mikkell{BV^{\Delta*}_t }&\mikkell{= \frac{N}{N-2K+2} \frac{1}{K \psi_K} \frac{\pi}{2}\sum_{j=0}^{N-2K+1} |r_{j,K}^*||r_{j+K,K}^*| - \frac{\hat{\omega}^2}{\theta^2 \psi_K}, } \label{eq:BVstar}
\end{align}
\mikkell{where $\psi_K := (1+2K^{-2})/12$ and $t \in [(i-1)\Delta,i\Delta]$. The term $\frac{\hat{\omega}^2}{\theta^2 \psi_K}$ is a bias-correction, where $\hat{\omega}^2$ is an estimate of the variance of the microstructure noise process $U$. We use the estimator of \cite{oomen06b}:}
\begin{align*}
\mikkell{\hat{\omega}_{AC}^2 = - \frac{1}{N-1} \sum_{j=2}^N r_ir_{i-1},}
\end{align*}
where $r_i = Z_i - Z_{i-1}$ is the \mikkol{$i$-th} log return. The \mikkol{statistics} $RV^{\Delta*}$ and $BV^{\Delta*}$ are the pre-averaged \mikkol{analogs} of the realized variance \citep{ABDL01} and bipower variation \citep{BNS04} estimators of IV, respectively. While both estimators are robust to market microstructure noise, only $BV^{\Delta*}$ is also robust to jumps in the price process. This is the main reason we use the $BV^{\Delta*}$ estimate of IV in this paper.

\section{Mathematical proofs}\label{app:proofs}

\subsection{Proofs of Propositions \ref{prop:BSSlm} and \ref{prop:BSSsm}}

Recall first that we can write
\begin{align*}
\rho_X(h) = \frac{\int_0^{\infty} g(x)g(x+\mikko{|h|})dx}{\int_0^{\infty}g(x)^2dx}, \quad \mikko{h \in \R}.
\end{align*}

\begin{proof}[Proof of Proposition \ref{prop:BSSlm}]
\eqref{asymp1} \mikkol{We may assume that} $h>1$, \mikkol{since we let $h \rightarrow \infty$}. By \mikkol{the assumption \eqref{A:infty}, we may} write 
\begin{align*}
\int_0^{\infty} g(x)g(x+h)dx &= \int_0^{\infty} g(x)(x+h)^{-\gamma} L_1(x+h)dx \\
	&= h^{-\gamma} L_1(h) \int_0^{\infty} g(x)(x/h+1)^{-\gamma} \frac{L_1(x+h)}{L(h)} dx \\
	&= h^{-\gamma} L_1(h) \int_0^{\infty} g(x)(x/h+1)^{-\gamma} \frac{L_1(h(x/h+1))}{L(h)} dx \\
	&\sim h^{-\gamma} L_1(h) \int_0^{\infty} g(x)dx, \quad h \rightarrow \infty, 
\end{align*}
\mikkol{using} the properties of slowly varying functions and where we applied the dominated convergence theorem, which is valid since for all $\epsilon > 0$ and large enough $h$,
\begin{align*}
g(x)\frac{L_1\left( (x/h+1)h\right)}{L_1(h)} < g(x)(1+\epsilon), \quad x \in(0,\infty),
\end{align*}
which is integrable over $(0,1)$ by assumption and over $[1,\infty)$ since $\gamma > 1$.

\eqref{asymp2} Let first
\begin{align*}
\int_0^{\infty} g(x)g(x+h)dx =& \int_0^1 g(x) (x+h)^{-\gamma}L_1(x+h) dx \\
&+ \int_1^{\infty} x^{-\gamma}(x+h)^{-\gamma}L_1(x) L_1(x+h)dx \\
	=:& I_{1,h} + I_{2,h}, 
\end{align*}
where
\begin{align*}
I_{1,h} &=  \int_0^1 g(x) (x+h)^{-\gamma}L_1(x+h) dx, \\
I_{2,h} &= \int_1^{\infty} x^{-\gamma}(x+h)^{-\gamma}L_1(x) L_1(x+h)dx.
\end{align*}
We may write
\begin{align*}
I_{2,h} & = h^{-\gamma} \int_{1}^\infty x^{-\gamma}\Big( 1+\frac{x}{h}\Big)^{-\gamma} L_1(x)  L_1(x+h) dx \\
& = h^{-2\gamma+1} L_1(h)^2 \int_{1/h}^\infty \underbrace{y^{-\gamma} (1+y)^{-\gamma} \frac{L_1(hy)}{L_1(h)}\frac{L_1(h(1+y))}{L_1(h)}}_{=: k_h(y)}dy,
\end{align*}
where the second equality follows by substituting $y = x/h$. Fix now $\delta \in (0,1-\gamma) \subset (0,1/2)$. By the Potter bounds \citep[][Theorem 1.5.6(ii)]{BGT}, under \eqref{A:infty} there exists a constant $C_\delta > 0$ such that
\begin{equation*}
\frac{L_1(hy)}{L_1(h)} \leq C_\delta \max \{ y^{-\delta},y^\delta\}, \quad y>1/h, \quad h>1.
\end{equation*}
Accordingly, we find a dominant for $k_h$, given by
\begin{equation*}
k_h(y) \leq \overline{k}(y) :=
\begin{cases}
C_\delta y^{-\gamma-\delta} (1+y)^{-\gamma+\delta}, & y \in (0,1],\\
C_\delta y^{-2(\gamma-\delta)}, & y \in (1,\infty),
\end{cases}
\end{equation*}
for any $h>1$. Note that the dominant $\overline{k}$ is integrable on $(0,1]$, since $-\gamma -\delta > -\gamma-(1-\gamma)=-1$, as well as on $(1,\infty)$, since $-2(\gamma-\delta) < -1$. Applying the dominated convergence theorem, we get
\begin{equation*}
I_{2,h} \sim h^{-2\gamma+1} L_1(h)^2 \int_0^\infty y^{-\gamma} (1+y)^{-\gamma} dy, \quad h \rightarrow \infty.
\end{equation*}
Moreover, the asymptotic estimate $I_{1,h} \sim c h^{-\gamma} L_1(h)$, as $h \rightarrow \infty$, can be deduced using the same strategy as  in \eqref{asymp1}. Observing that $-\gamma < -2\gamma+1$ when $\gamma < 1$, we deduce that $I_{1,h} = o\big(h^{-2\gamma+1} L_1(h)^2\big)$ as $h \rightarrow \infty$\mikkol{, so the} assertion follows.
\end{proof}

\begin{proof}[Proof of Proposition \ref{prop:BSSsm}]
As in the proof of Proposition \ref{prop:BSSlm}\mikko{,} we take $h > 1$ and write
\begin{align*}
\int_0^{\infty} g(x)g(x+h)dx &= \int_0^{\infty} g(x)(x+h)^{-\gamma}e^{-\lambda(x+h)} L_1(x+h)dx \\
 &= e^{-\lambda h}h^{-\gamma}L_1(h) \int_0^{\infty} g(x)(x/h+1)^{-\gamma}e^{-\lambda x} \frac{L_1(x+h)}{L(h)}dx  \\
 &\sim e^{-\lambda h}h^{-\gamma}L_1(h) \int_0^{\infty} g(x)e^{-\lambda x}dx, 
\end{align*}
where we used the dominated convergence theorem, which can be justified in the same manner as in \mikkol{the proof of} Proposition \ref{prop:BSSlm}.
\end{proof}

\subsection{Proof of Theorem \ref{th:sig}}
\mikkol{We first recall an elementary result that we will need below. Namely, as $x \downarrow 0$,
\begin{align}\label{eq:exp}
\sum_{k=2}^{\infty} \frac{x^k}{k!} = o(x).
\end{align}
This result is easy to prove by examining the Taylor expansion of the exponential function and applying l'H\^opital's rule.}

\begin{proof}[Proof of Theorem \ref{th:sig}]
\eqref{item:rough} Suppose w.l.o.g.\ that $\xi =1$. Since $\sigma$ is \mikko{covariance-stationary,} \mikko{it suffices to study}
\begin{align*}
\rho(h) = Corr(\sigma_t,\sigma_{t+h}) = \frac{Cov(\sigma_t,\sigma_{t+h})}{Var(\sigma_0)}, \quad h \geq 0.
\end{align*}
We get using the definition \eqref{eq:sig} of $\sigma$
\begin{align*}
Cov(\sigma_t,\sigma_{t+h}) &= \E[\sigma_t \sigma_{t+h}] - \E[\sigma_0]^2 \\
	&=\E[\exp(X_t + X_{t+h})] -  \E[\exp(X_0)]^2.
\end{align*}
Now, by the fact that $X$ is a zero-mean Gaussian process we get, using the moment generating function of the Gaussian distribution,
\begin{align*}
\E[\exp(X_t + X_{t+h})] &= \exp \left( \frac{1}{2} Var(X_t + X_{t+h}) \right) \\
	&= \exp \left( Var(X_0) + cov(X_0,X_{t+h}) \right) \\
	&= \exp \left( \gamma_X(0) + \gamma_X(h) \right),
\end{align*}
and
\begin{align*}
\E[\exp(X_0)] = \exp\left(\frac{1}{2} \gamma_X(0)\right),
\end{align*}
where we write $\gamma_X(h)$ for $Cov(X_h,X_0) = \E[X_h X_0]$. Putting this together, we arrive at
\begin{align*}
\rho(h) = \frac{ \exp(\gamma_X(h)) - 1}{\exp(\gamma(0)) - 1}.
\end{align*}
Now, \mikko{using generic positive constants} $C, C_1, C_2$ \mikko{that may vary} from line to line\mikko{,}
\begin{align*}
1-\rho(h) &= \frac{\exp(\gamma_X(0)) -1 - \left(\exp(\gamma_X(h)) - 1\right) }{\exp(\gamma(0)) -1} \\
	&= \frac{\exp(\gamma_X(0)) -\exp(\gamma_X(h)) }{\exp(\gamma(0)) -1} \\
	&= C \left[ 1- \exp (\gamma_X(h) - \gamma_X(0)) \right] \\
	&=  C \left[ 1- \exp (-\gamma_X(0) (1- \rho_X(h))) \right] \\
	&=   C_1 (1-\rho_X(h)) + C_2 \sum_{k=2}^{\infty} \frac{\left( -\gamma_X(0)(1-\rho_X(h))\right)^k}{k!} \\
	&\sim C |h|^{2\alpha+1}L_0(h),
\end{align*} 
by the assumption on $1-\rho_X(h)$ \mikkol{and the asymptotic result} \eqref{eq:exp}\mikko{.} \mikko{(O}n the \mikko{penultimate} line\mikko{,} we Taylor expanded the exponential.\mikko{)} This concludes the proof of \eqref{item:rough}.

\eqref{item:memory} Suppose w.l.o.g.\ that $\xi =1$. Using the same approach as in part \eqref{item:rough}, we get by Taylor expansion:
\begin{align*}
\rho(h) &= \frac{\exp(\gamma_X(h)) - 1 }{\exp(\gamma(0)) -1} \\
	&= C_1 \gamma_X(h) + C_2 \sum_{k=2}^{\infty} \frac{\gamma_X(h)^k}{k!} \\
	&= C_1 \gamma_X(0) \rho_X(h) + C_2 \sum_{k=2}^{\infty} \frac{\gamma_X(h)^k}{k!} \\
	&\sim C \rho_X(h), \quad h \rightarrow \infty,
\end{align*}
\mikkol{where we have finally applied} \eqref{eq:exp}.
\end{proof}

\subsection{Proof of Theorem \ref{th:NLLS}}

\begin{proof}[Proof of Theorem \ref{th:NLLS}]
Let
\begin{align*}
Q_n(\theta) = Q_n(a,b,\alpha) := \sum_{k=1}^m \left( \hat{\gamma}_2^*(k\Delta) - a - b (k\Delta)^{2\alpha+1} \right)^2,
\end{align*}
and
\begin{align*}
Q(\theta) = Q(a,b,\alpha) := \sum_{k=1}^m \left( a_0 + b_0 (k\Delta)^{2\alpha_0+1} - a - b (k\Delta)^{2\alpha+1} \right)^2,
\end{align*}
where $a_0 := 2\sigma_u^2$. Since, by the assumptions on $X$ and $Z$, $\hat{\gamma}_2^*(k\Delta)$ is a consistent estimator of $a_0 + b_0 (k\Delta)^{2\alpha_0+1}$, it is not hard to show that
\begin{align}\label{eq:ucon}
\sup_{\theta \in \Theta} | Q_n(\theta) -Q(\theta)| \rightarrow 0, \quad n \rightarrow \infty.
\end{align}
Because $m\geq 3$, $\theta_0 := (a_0,b_0,\alpha_0)$ is the unique minimizer of $Q(\theta)$ (in fact $Q(\theta_0) = 0$)\mikkol{,} the rest of the proof therefore follows standard arguments, which we give here for completeness. 

Let $\epsilon > 0$. For $n$ sufficiently large, we have\mikkol{,} because of the uniform convergence \eqref{eq:ucon}, that
\begin{align*}
|Q(\hat{\theta}) - Q_n(\hat{\theta})| < \epsilon.
\end{align*}
Using first that $Q(\theta_0) = 0$, then that $\hat{\theta}$ minimizes $Q_n(\theta)$, and finally the uniform convergence \eqref{eq:ucon} again, it holds that
\begin{align*}
|Q_n(\hat{\theta}) - Q(\theta_0)| = Q_n(\hat{\theta}) \leq Q_n(\theta_0) <  \epsilon,
\end{align*}
for $n$ sufficiently large. Let $B$ be an open neighborhood of $\theta_0$. Now\mikkol{,}
\begin{align*}
\mathbb{P}(\hat{\theta} \in B^c \cap \Theta) &\leq \mathbb{P}( | Q(\hat{\theta}) - Q(\theta_0)| > 0) \\\
	&\leq \mathbb{P}( |Q(\hat{\theta}) - Q_n(\hat{\theta})| + |Q_n(\hat{\theta}) - Q(\theta_0)| > 0) \\
	& \rightarrow 0, \quad n \rightarrow \infty,
\end{align*}
by the results above\mikkol{,} since $\epsilon>0$ was arbitrary\mikkol{.} 
\end{proof}

\end{document}